\pdfoutput=1
\documentclass[sigconf]{acmart}
\usepackage[linesnumbered,ruled,vlined]{algorithm2e}
\usepackage{pifont}
\usepackage{svg} %
\usepackage{float}
\usepackage{graphicx} 
\usepackage{caption}
\usepackage{subcaption} 

\newcommand{\CHENG}{}
\newcommand{\JIANYANG}{}
\newcommand{\CHENGB}{}
\newcommand{\JIANYANGB}{}
\newcommand{\JIANYANGLAST}{}
\newcommand{\CHENGC}{}
\newcommand{\JIANYANGREVISION}{}
\newcommand{\chengr}{}
\newcommand{\JIANYANGCAMERA}{}
\newcommand{\chengf}{}

\AtBeginDocument{%
  \providecommand\BibTeX{{%
    \normalfont B\kern-0.5em{\scshape i\kern-0.25em b}\kern-0.8em\TeX}}}

\setcopyright{acmcopyright}
\copyrightyear{2023}
\acmYear{2023}
\acmDOI{XXXXXXX.XXXXXXX}

\acmConference[SIGMOD '23]{the 2023
International Conference on Management of Data}{June 18--23,
  2023}{Seattle, WA, USA}
\acmPrice{15.00}
\acmISBN{978-1-4503-XXXX-X/18/06}

\begin{document}
\sloppy


\title{High-Dimensional Approximate Nearest Neighbor Search: with Reliable and Efficient Distance Comparison Operations}

\author{Jianyang Gao}
\affiliation{%
  \institution{Nanyang Technological University}
  \country{Singapore}}
\email{jianyang.gao@ntu.edu.sg}

\author{Cheng Long}
\affiliation{%
  \institution{Nanyang Technological University}
  \country{Singapore}}
\email{c.long@ntu.edu.sg}
\renewcommand{\shortauthors}{Jianyang Gao and Cheng Long}

\begin{abstract}
Approximate K nearest neighbor (AKNN) search {\JIANYANGREVISION in the high-dimensional Euclidean vector space} is a fundamental and challenging problem. {\CHENGB We observe that} in high-dimensional space, the time consumption of {\CHENGB nearly all} AKNN algorithms is dominated by {\CHENGB that of} the distance comparison operations (DCOs). For each operation, it scans full dimensions of an object and thus, runs in linear time wrt the dimensionality.
To speed it up, we propose a randomized algorithm named \texttt{ADSampling} which runs in logarithmic time {\CHENGC wrt} the dimensionality {\CHENGB for the majority of DCOs} and succeeds with 
high probability. 
In addition, based on \texttt{ADSampling} we develop one {\chengf generic} and two algorithm-specific techniques as plugins to enhance existing AKNN algorithms.
Both theoretical and empirical studies confirm that: (1) our techniques introduce nearly no accuracy loss and (2) they consistently improve the efficiency.
\end{abstract}

\maketitle

\section{Introduction}
\label{sec:introduction}


K nearest neighbor (KNN) search in the high-dimensional Euclidean {\JIANYANGREVISION vector} space
is a fundamental problem and has {\CHENGB a wide range of} applications in information retrieval~\cite{KNNimageretrieval}, data mining~\cite{KNNClassficiation} and recommendations~\cite{KNNrecommendation}. However, due to the curse of dimensionality~\cite{indyk1998approximate}, exact KNN query usually requires unacceptable response time. To achieve better time and accuracy tradeoff, many researchers turn to its relaxed version, {\CHENGB namely} approximate K nearest neighbor (AKNN) search~\cite{datar2004locality, muja2014scalable, jegou2010product, malkov2018efficient, ge2013optimized, guo2020accelerating}.

{\CHENG Many algorithms have been developed for AKNN, including 
{\JIANYANGREVISION (1) graph-based~\cite{malkov2018efficient, NSW, li2019approximate, fu2019fast, fu2021high, SISAP_graph}}, (2) quantization-based~\cite{jegou2010product, ge2013optimized, guo2020accelerating, ITQ, additivePQ, imi}, (3) {\JIANYANGREVISION tree-based ~\cite{muja2014scalable, dasgupta2008random, ram2019revisiting, beygelzimer2006cover, reviewer_M_tree}} and (4) hashing-based~\cite{indyk1998approximate, datar2004locality, c2lsh, tao2010efficient, huang2015query, sun2014srs, lu2020vhp, zheng2020pm, james_cheng}.
{\JIANYANGREVISION In particular, we focus on in-memory AKNN algorithms which assume that all raw data vectors and indexes can be hosted in the main memory~\cite{malkov2018efficient, li2019approximate, fu2019fast,NSW, jegou2010product, muja2014scalable, beygelzimer2006cover, zheng2020pm}.}
These algorithms all adopt the strategy of first generating {\JIANYANGLAST some} candidates {\JIANYANGLAST for} KNNs and then finding {\JIANYANGLAST out} the KNNs among 
{\JIANYANGLAST them}.~\footnote{Graph-based methods generate candidates and find {\JIANYANGLAST out} the KNNs among the candidates generated so far \emph{iteratively}.} \underline{First}, they differ in their ways of generating candidates of KNNs. 
For example, graph-based methods organize the vectors with a graph and conduct a heuristic-based search (e.g., greedy search) on the graph for generating candidates.
\underline{Second}, these algorithms largely share their ways of finding KNNs among the candidates. 
Specifically, they maintain a KNN set $\mathcal Q$~\footnote{In graph-based methods, the size of $\mathcal Q$ is set to be an integer $N_{ef} > K$ since the the distance of the $N_{ef}$th 
NN is required for generating candidates. In other AKNN methods, the size of $\mathcal Q$ is set to be $K$.} (technically, a max-heap), and for a new candidate, they check whether its distance {\JIANYANGLAST {\CHENGC from the query}} is 
{\JIANYANGLAST no greater than}
the maximum in $\mathcal Q$. If so, they include the candidate to $\mathcal Q$ with the distance as a key;~\footnote{{\CHENGB {\JIANYANGCAMERA When} the distance is equal to the maximum in $\mathcal{Q}$, {\chengf they} can choose not to include {\JIANYANGCAMERA it.}}} otherwise, the candidate is discarded. 
We call the computation of \textbf{checking whether an object has its distance from a query 
{\JIANYANGLAST no greater than}
a distance threshold and providing its distance if so} a \emph{distance comparison operation} (DCO).
Given an object $\mathbf{o}$ and a distance threshold $r$, we say that $\mathbf{o}$ is a \emph{positive} object (wrt $r$) if $\mathbf{o}$'s distance from the query is at most $r$ and a \emph{negative} object otherwise. 


All existing AKNN algorithms adopt the following method for the DCO for an object and a threshold. It first computes the object's distance (from the query) and then compares the computed distance against the threshold. We call this method \texttt{FDScanning} since it scans full dimensions of the object for the operation. Clearly, \texttt{FDScanning} has the time complexity of $O(D)$, where $D$ is the number of dimensions of an object. 
Based on \texttt{FDScanning}, nearly all AKNN algorithms have their time costs dominated by that of performing DCOs. We consider one of the most popular AKNN algorithms, \texttt{HNSW}~\cite{malkov2018efficient}, 
for illustration (elaborations on other algorithms will be provided in Section~\ref{subsec:aknn}). 
Let $N_s$ be the number of generated candidates of KNNs. The total time cost of \texttt{HNSW} is $O(N_s D + N_s \log N_s)$, where $O(N_s D)$ is the cost of performing DCOs and $O(N_s \log N_s)$ is the cost of other computations (detailed analysis can be found in Section~\ref{subsec:aknn}). 
Since $D$ can be 
{\JIANYANGREVISION hundreds}
while $\log N_s$ is a few dozens only for a big dataset involving millions of objects in practice, the cost of performing DCOs dominates the total cost of \texttt{HNSW} {\JIANYANG (we empirically verify the statement in Section~\ref{subsec:aknn}.)}. {\JIANYANG For example, on {\CHENG a real dataset DEEP, which involves 256} dimensions, the DCOs take 77.2\% of the total running time.}

We have two observations.
%
\underline{First}, for {\CHENGB DCOs on} negative objects, we only need to confirm that the objects' distances are larger than the threshold distance \emph{without returning their exact distances} - recall that negative objects would be discarded.
{\CHENGB Therefore, \texttt{FDScanning}, which always computes the distance of an object for a DCO on the object, conducts more than necessary computation for negative objects.}
\underline{Second}, among the DCOs involved in an AKNN algorithm, often most would be for negative objects. A clue for this is that when the DCO is conducted on an object, the distance threshold corresponds to some small distance (e.g., the $K$th smallest distance seen so far in many AKNN algorithms). As a result, the object would likely have its distance larger than the threshold and correspond to a negative object. We verify this statement empirically (details can be found in Section~\ref{subsec:aknn}). For example, 
for a representative 
algorithm \texttt{IVF}~\cite{jegou2010product}, the number of negative objects is 60x to 869x more than that of the positive ones.
These observations collectively show that \texttt{FDScanning} is an over-kill for most of the DCOs that are involved for answering a KNN query, and thus there exists much room to achieve 
{\CHENGB reliable (nearly-exact) DCOs{\CHENGB~\footnote{We target reliable DCOs since DCOs are used for finding KNNs among the generated candidates, and if their accuracy is compromised, the quality of the found KNNs would be largely affected.}} with better efficiency.


In the literature, no efforts have been devoted to achieving 
reliable DCOs with better efficiency
than \texttt{FDScanning}, to the best of our knowledge. 
An immediate attempt is to use some distance approximation techniques such as \emph{product quantization}~\cite{ge2013optimized, jegou2010product,ITQ} and \emph{random projection}~\cite{johnson1984extensions} for DCOs in order to achieve better efficiency. However, as widely found in the literature~\cite{learningtohashsurvey, adaptive2020ml} and also empirically verified in our experiments, these techniques cannot avoid accuracy sacrifice in order to achieve some remarkable time cost savings, and thus they can hardly be used to achieve \emph{reliable} DCOs with better efficiency. For example, {\JIANYANG according to \cite{learningtohashsurvey}, on the dataset SIFT1M with one million 128-dimensional vectors, none of the 
{\JIANYANGLAST quantization}
algorithms achieve more than 60\% recall without re-ranking.} Indeed, these techniques have only been used for generating candidates of KNNs~\cite{learningtohashsurvey}, but not (in DCOs) for 
{\JIANYANGB finding them out from the generated candidates}.}

Therefore, in this paper, we develop a new method called \texttt{ADSampling} to fulfill this goal. 
At its core, \texttt{ADSampling} projects the objects to 
spaces with different dimensionalities
and conduct DCOs based on the projected vectors for better efficiency.
Different from the conventional 
random projection technique~\cite{johnson1984extensions, datar2004locality, sun2014srs, c2lsh}, which projects \emph{all} objects to vectors with \emph{equal} dimensions, \texttt{ADSampling} is novel in the following aspects. \underline{First}, it projects \emph{different} objects to vectors with \emph{different} 
numbers of dimensions
{\CHENGB during the query phase} \emph{flexibly}. The rationale is that for {\JIANYANG negative} objects that are farther away from the query, it would be sufficient to project them to a space with fewer dimensions for {\CHENGC reliable} DCOs; whereas for {\JIANYANG negative} objects that are closer to the query, they should be projected to a space with more dimensions for {\CHENGC reliable} DCOs. We note that for positive objects (i.e., those that have their distances from the query at most a threshold), their distances should ideally not be distorted. \texttt{ADSampling} achieves this flexibility with two steps. (1) It first {\JIANYANG preprocesses} the objects via a \emph{random orthogonal transformation}~\cite{choromanski2017unreasonable, randomortho} {\CHENGB during the index phase (i.e., before a query comes)}. This step merely {\JIANYANG randomly} rotates the objects without distorting their distances. 
{\JIANYANG 
(2) Then during the query phase, when handling DCOs {\CHENGB on different objects}, it samples different numbers of dimensions of their transformed vectors. {\CHENGB We verify that the sampled vectors produced by these two steps 
{\JIANYANG are identically distributed with}
those} 
obtained from random projection, and thus, the approximate distances based on the sampled vectors, like those based on random projection, {\CHENGB correspond to} good estimations of the true distances while {\CHENGB achieving} the aforementioned flexibility of dimensionality.} 

\underline{Second}, {\CHENGB it decides the number of dimensions to be sampled for each object \emph{adaptively} based on 
{\CHENGC the DCO on the object}
during the query phase, but not pre-sets it to a certain number during the index phase (which is knowledge demanding and difficult to set in practice).}
Specifically, it 
{\JIANYANGREVISION \emph{incrementally}}
samples the dimensions of a transformed vector until it can confidently conduct the DCO on the object based on 
{\JIANYANGLAST the sampled vector}. 
With 
{\JIANYANGLAST the sampled vector}, 
{\CHENGC it determines whether there has been enough evidence for a reliable DCO by computing an approximate distance of the object and then conducting a \emph{hypothesis testing} based on the computed approximate distance.}
This is possible due to 
{\JIANYANGB the fact that there is a theoretical error bound on the approximate distance}
{\CHENGB(recall the aforementioned equivalence between a projected vector by \texttt{ADSampling} and that by the conventional random projection)}.

{\CHENGB \texttt{ADSampling} achieves reliable DCOs with better efficiency
than \texttt{FDScanning}, which we explain as follows.}
First, for each negative object, {\JIANYANG we prove that} \texttt{ADSampling} would always return the correct answer 
and run in \emph{logarithmic} time wrt $D$ {\JIANYANG in expectation}
{\JIANYANG (recall that \texttt{FDSanning} runs in linear time wrt $D$). }
Second, {\JIANYANG for each positive object, it succeeds with high probability and runs in $O(D)$ {\CHENGC time}.}
Third, there are much more negative objects than positive objects.

We summarize the major contributions of this paper as follows.
\begin{enumerate}
\item We systematically review existing AKNN algorithms and identify the \emph{distance comparison operation} (DCO), which is ubiquitous in AKNN algorithms. With the existing method \texttt{FDScanning}, the costs of DCOs dominate the overall costs for nearly all AKNN algorithms, which we verify both theoretically and empirically. (Section~\ref{sec:dco})

\item We propose a new method \texttt{ADSampling}, 
which achieves 
{\CHENGB reliable DCOs with better efficiency}
{\JIANYANGREVISION for the high-dimensional Euclidean space}. 
Specifically, in most of the cases (i.e., for negative objects), \texttt{ADSampling} runs in \emph{logarithmic} time wrt $D$ and always returns a correct answer. (Section~\ref{sec:adsampling})

\item For a general AKNN algorithm (which we denote by \texttt{AKNN}), we replace \texttt{FDScanning} with \texttt{ADSampling} for DCOs and achieve a new algorithm (which we denote by \texttt{AKNN+}). 
{\JIANYANG We prove that an \texttt{AKNN+} algorithm preserves the results of its corresponding \texttt{AKNN} algorithm with high probability and significantly reduces the time complexity. (Section~\ref{sec:aknn+})}

\item 
{\JIANYANGCAMERA We further develop two AKNN-algorithm-specific techniques to improve the cost-effectiveness of \texttt{AKNN+} algorithms.}
For example, for 
{\JIANYANG graph-based algorithms (with \texttt{HNSW+} as a representative)}, we incorporate more approximations and obtain a new algorithm (which we call \texttt{HNSW++}); for other algorithms (with \texttt{IVF+} as a representative), we improve their cost-effectiveness with cache-friendly storage. (Section~\ref{sec:aknn++})

\item We conduct extensive experiments on real datasets, which verify our techniques.
{\JIANYANG
For example, \texttt{ADSampling} brings up to 2.65x speed-up on \texttt{HNSW} and 5.58x on \texttt{IVF} {\CHENG while providing} the same accuracy. 
Besides, 
it helps to save up to 75.3\% of the evaluated dimensions for \texttt{HNSW} and up to 89.2\% of those for \texttt{IVF} with the accuracy loss of no more than 0.14\%. (Section~\ref{section:experiment})
}


\end{enumerate}
%
{\JIANYANGREVISION For the rest of the paper, we review the related work in Section~\ref{sec:related work} and conclude the paper in Section~\ref{sec:conclusion}.}
}

\section{The Distance Comparison Operation}
\label{sec:dco}

\subsection{KNN Query and Distance Comparison Operation}
\label{subsec:definition}

Let $\mathcal O$ be a database of $N$ objects in a $D$-dimensional Euclidean space $\mathbb{R}^D$ and $\mathbf{q}$ be a query. {\CHENG In this paper, we use ``object'' (resp. ``query'') and ``data vector'' (resp. ``query vector'') interchangeably}. 
{\CHENG We note that operations on $\mathcal{O}$ can be conducted before the query $\mathbf{q}$ comes (i.e., the index phase) while those on the query $\mathbf{q}$ can only be conducted after it comes (i.e., the query phase). 
For an object $\mathbf{o}$, we define its difference from the query $\mathbf{q}$ as $\mathbf{x}$, i.e., $\mathbf{x} = \mathbf{o} - \mathbf{q}$. We refer to an object $\mathbf{o}$ by its corresponding vector $\mathbf{x}$ when the context is clear.}
Without ambiguity, by ``the distance of an object $\mathbf{o}$'', we refer to its distance {\CHENGC from} the query vector $\mathbf{q}$, {\CHENG which we denote by $dis$}. 
%
%
The \textbf{K nearest neighbor (KNN)} query is to find the top-K objects with the minimum distance {\CHENGC from} the query $\mathbf{q}$. 


{\CHENG In this paper, we study the \emph{distance comparison operation} (DCO), which is defined as follows.}
\begin{definition}[Distance Comparison Operation]
{\CHENG Given an object $\mathbf{o}$ and a distance threshold $r$, the \textbf{distance comparison operation} (DCO) is to decide whether object $\mathbf{o}$ has its distance $dis$ {\JIANYANGLAST no greater than $r$}
and if so, return $dis$.
In particular, we say object $\mathbf{o}$ is a \emph{positive} object if {\CHENGC $dis\le r$} and a \emph{negative} object otherwise.}
\end{definition}
%

{\CHENG As mentioned in Section~\ref{sec:introduction}, DCOs are heavily involved in many AKNN algorithms. These algorithms conduct the DCO for an object $\mathbf{o}$ and a distance $r$ naturally by computing $\mathbf{o}$'s distance and comparing the distance against $r$. We call this conventional method \texttt{FDScanning} since it uses \emph{all} dimensions of $\mathbf{o}$ {\CHENGC for computing} the distance. Clearly, \texttt{FDScanning} has the time complexity of $O(D)$. 

Next, we review the existing AKNN algorithms and validate both theoretically and empirically the critical role of DCOs in these algorithms.}

\subsection{AKNN Algorithms and Their DCOs}
\label{subsec:aknn}

\begin{figure}[thb]
    \vspace{-4mm}
    \centering
    \begin{subfigure}[b]{0.5\linewidth}
        \centering
        \includegraphics[width=0.8\textwidth]{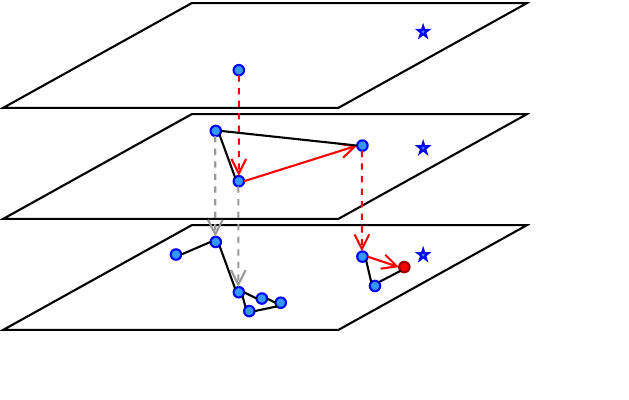}
        \caption{\texttt{HNSW}}
        \label{fig:routing}
    \end{subfigure} 
    \begin{subfigure}[b]{0.45\linewidth}
        \centering
        \includegraphics[width=0.8\textwidth]{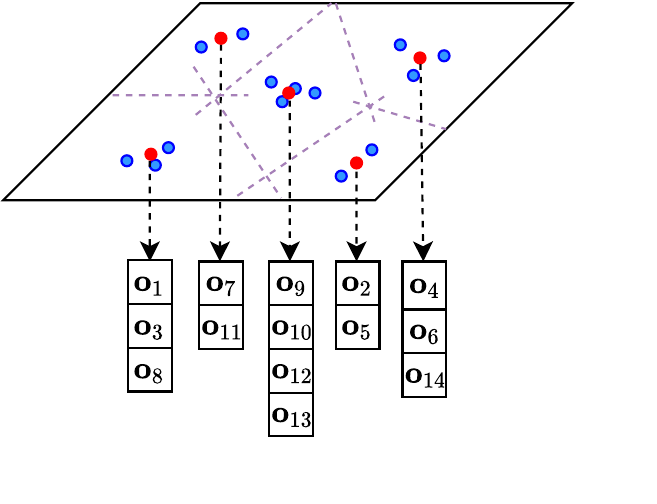}
        \caption{\texttt{IVF}}
	\label{fig:ivf}
    \end{subfigure} 
    \vspace{-4mm}
    \caption{Illustrations of AKNN algorithms.}
    \vspace{-4mm}
    \label{fig:aknn}
\end{figure}

\subsubsection{Graph-Based Methods}
{\JIANYANG Graph-based methods are one family of state-of-the-art AKNN algorithms that exhibit dominant performance on the time-accuracy tradeoff for 
{\JIANYANGREVISION in-memory AKNN query~\cite{malkov2018efficient, NSW, li2019approximate, fu2019fast, fu2021high, SISAP_graph}}.
{\CHENGC These methods construct graphs based on the data vectors, where} a vertex corresponds to a data vector. }
One famous graph-based method is the hierarchical navigable small world graphs (\texttt{HNSW})~\cite{malkov2018efficient}. It's composed of several layers. Layer 0 (base layer) contains all data vectors and layer $i+1$ only keeps a subset of the vectors in layer $i$ randomly. 
The size of each layer decays exponentially as it goes up.
In particular, the top layer contains only one vertex. Within each layer, a vertex is connected to its several approximate nearest neighbors,
{\JIANYANG while between adjacent layers, two vertexes are connected only if they represent the same vector. }
{\CHENG An illustration of the \texttt{HNSW} graph is provided in Figure~\ref{fig:routing}.}

During the query phase, greedy search is first performed on upper layers to find a good entry at layer 0 (the base layer). 
Specifically, the search starts from the only vertex of the top layer. 
{\JIANYANG Within each layer, it does greedy search iteratively. At each iteration, it accesses all the neighbors of its currently located vertex and goes to the one with the minimum distance. It terminates the search when {\CHENG none of the neighbors has a smaller distance than the currently located vertex.}} Then it goes to the next layer and repeats the process until {\CHENG it arrives at} layer 0. 
{\CHENG At layer 0, it conducts \textit{greedy beam search}~\cite{graphbenchmark} (\textit{best first search}), which is adopted by 
most graph-based methods~\cite{malkov2018efficient, li2019approximate, fu2019fast, diskann, NSW}.
}
To be specific, {\CHENG greedy beam search maintains two sets}: a search set $\mathcal S$ (a min-heap by exact distances) and a result set $\mathcal R$ (a max-heap by exact distances). The search set $\mathcal S$ {\CHENG has its size unbounded and} maintains candidates yet to be searched. The result set $\mathcal R$ {\CHENG has its size bounded by $N_{ef}$ and} maintains $N_{ef}$ nearest neighbors {\CHENG visited so far}, where the size $N_{ef}$ is the parameter to control time-accuracy trade-off. 
At the beginning, a start point {\CHENG at layer 0} is inserted into both $\mathcal S$ and $\mathcal R$. Then {\CHENG it proceeds in iterations. At each iteration, it pops 
{\JIANYANG the object with the smallest distance in set $\mathcal{S}$}
and enumerates the neighbors of the object. 
For each neighbor, it \textbf{checks whether its distance from the query object is 
{\JIANYANGLAST no greater than}
the maximum distance in set $\mathcal{R}$ and if so, it computes the distance} (i.e., it conducts a DCO).
In addition, 
{\CHENGB if the distance is smaller than the maximum distance in $\mathcal{R}$,}
it (1) pushes the object {\JIANYANGLAST into} both set $\mathcal{S}$ and set $\mathcal{R}$ (using the computed distance as the key)
and (2) pops {\CHENGC the object with the maximum distance} from set $\mathcal{R}$ whenever $\mathcal{R}$ involves more than $N_{ef}$ objects so that the size of $\mathcal{R}$ is bounded by $N_{ef}$.
It returns $K$ objects in $\mathcal{R}$ with the smallest distances when the minimum distance in $\mathcal S$ becomes larger than the maximum distance in $\mathcal R$ and stops.
%
We note that the greedy search at upper layers corresponds to a greedy beam search process with $N_{ef}=1$.

\begin{figure}[thb]
    \captionsetup[subfigure]{aboveskip=-3pt}
    \centering
    \vspace{-4mm}
    \begin{subfigure}[b]{0.32\linewidth}
        \includegraphics[width=\textwidth]{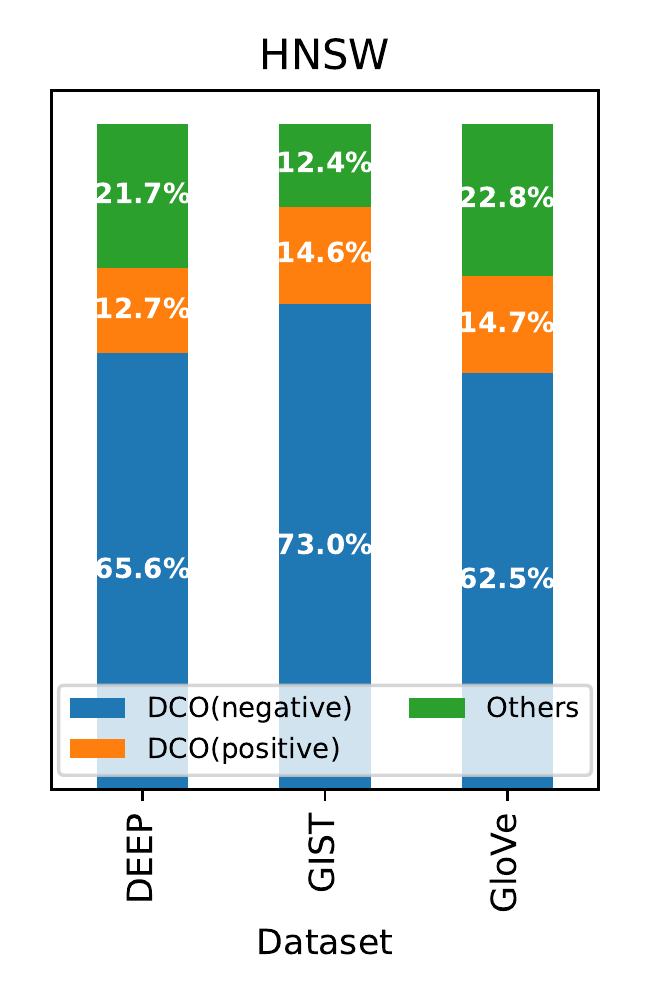}
        \caption{\texttt{HNSW}}
        \label{fig:cost HNSW}
    \end{subfigure} 
    \begin{subfigure}[b]{0.32\linewidth}
        \includegraphics[width=\textwidth]{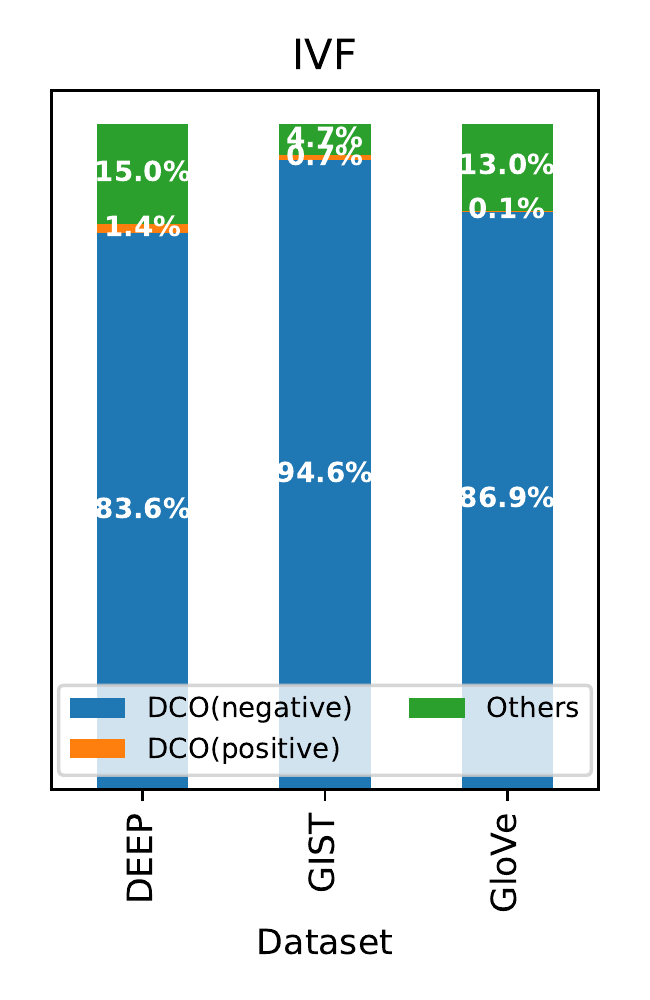}
        \caption{\texttt{IVF}}
	  \label{fig:cost IVF}
    \end{subfigure} 
    \begin{subfigure}[b]{0.32\linewidth}
        \includegraphics[width=\textwidth]{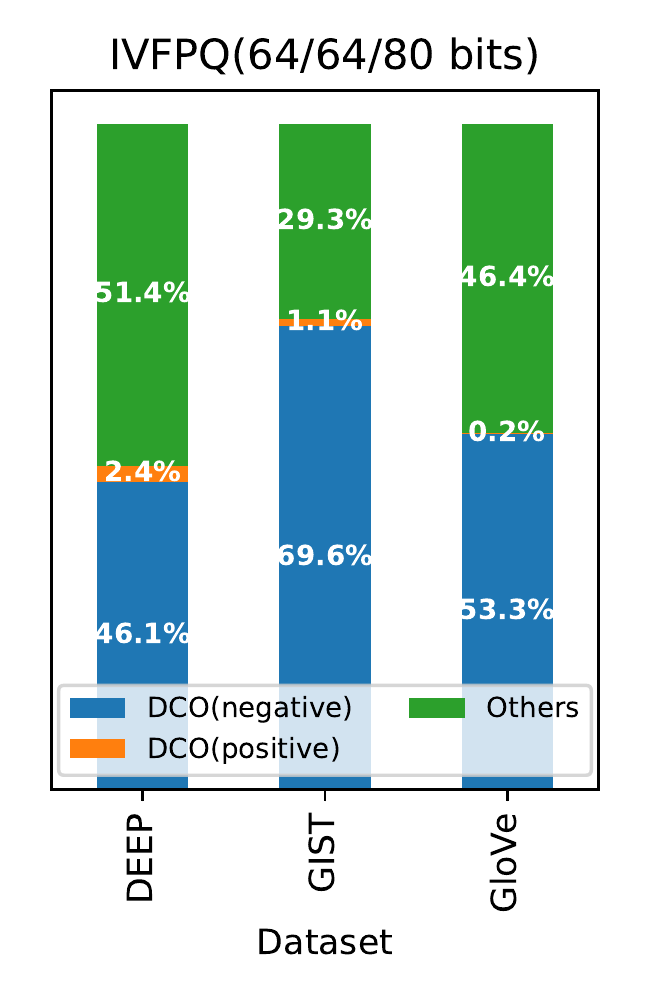}
        \caption{\texttt{IVFPQ}}
	  \label{fig:cost IVFPQ}
    \end{subfigure} 
    \vspace{-4mm}
    \caption{{\JIANYANGREVISION {\CHENG Breakdown of Running Times of} AKNN Algorithms.}}
    \vspace{-4mm}
    \label{fig:cost statistics}
\end{figure}

\smallskip
\noindent\textbf{DCO v.s. Overall Time Costs.} We review the time complexity of \texttt{HNSW} 
assuming {\JIANYANGLAST that} it adopts \texttt{FDScanning} for DCOs. Let $N_s$ be the number of {\JIANYANGLAST the} candidates of KNN objects, which are visited by \texttt{HNSW}.
Then, the total cost of the DCOs is $O(N_s D)$ and that of updating the sets $\mathcal{S} $ and $\mathcal{R} $ is $O(N_s \log N_s)$. Therefore, the time complexity of \texttt{HNSW} is $O(N_s D + N_s \log N_s)$. In practice, the total cost of DCOs should be the dominating part since $D$ can be 
{\JIANYANGREVISION hundreds}
while $\log N_s$ is a few dozens only for a big dataset involving millions of objects.
We verify this empirically as well. 
Figure~\ref{fig:cost HNSW} profiles the time consumption of \texttt{HNSW} on three real-world datasets when targeting 95\% recall with $K = 100$. 
According to the results, on datasets with various dimensions from 256 to 960, DCOs take from 77.2\% to 87.6\% of the total running time of \texttt{HNSW} (as indicated by the blue and orange portions of the bars). 
%

\smallskip
\noindent\textbf{Positive v.s. Negative Objects.} We verify empirically that for \texttt{HNSW}, the number of DCOs on negative objects is significantly larger than that of DCOs on positive objects. 
The results are shown in Figure~\ref{fig:cost HNSW}. We note that in the figure, the ratio between the cost of DCOs (on  negative objects) and that (on positive objects) reflects the ratio between the numbers of negative and positive objects since a DCO on a negative object and that on a positive object have the same cost. 
According to the results, the number of negative objects is 4.3x to 5.2x times more than that of the positive ones.
}




\subsubsection{Inverted File Index}

Inverted file~\cite{jegou2010product} index is another popular index method for AKNN query. {\CHENG According to~\cite{adaptive2020ml}, \texttt{IVF} is one of the state-of-the-art approaches for AKNN. Indeed, according to our experimental results in Section~\ref{section:time-accuracy}, it outperforms \texttt{HNSW} on some datasets. 
} 
During the index phase, the algorithm clusters data vectors with the K-means algorithm, builds a bucket for each cluster and assigns each data vector to its corresponding bucket. Then during the query phase, for a given query, the algorithm first selects the $N_{probe}$ nearest clusters based on their centroids, retrieves all vectors in these corresponding buckets as candidates, and then 
finds out KNNs {\CHENG among the retrieved vectors.}
Here,
$N_{probe}$ is a user parameter which controls 
the time-accuracy trade-off. 
{\CHENG When {\JIANYANGLAST finding out} KNNs, a commonly used method is to maintain a KNN set $\mathcal K$ with a max-heap of size $K$. It then scans all candidates, and for each one, it \textbf{checks whether its distance is 
{\JIANYANGLAST no greater than}
the maximum of $\mathcal{K}$ and if so, it computes the distance} (i.e., it conducts a DCO). Here, the maximum distance is defined to be $+\infty$ if $\mathcal{K}$ is not full.
{\CHENGB If the distance is smaller than the maximum distance in $\mathcal{K}$,}
it 
{\JIANYANG updates $\mathcal{K}$ with the candidate}
(by using the computed distance as the key). It returns the objects in $\mathcal{K}$ at the end. An illustration of the \texttt{IVF} structure is provided in Figure~\ref{fig:ivf}.

\smallskip
\noindent\textbf{DCO v.s. Overall Time Costs.} We review the time complexity of \texttt{IVF}
assuming {\JIANYANGLAST that} it adopts \texttt{FDScanning} for DCOs.
Let $N_s$ be the number of candidate objects. The total cost of \texttt{IVF} is $O(N_s D + N_s\log K)$, where the first term is the cost of the DCOs and the second term is that of updating $\mathcal K$. As can be noticed, the cost of DCOs is the dominating part. }
We verify this empirically as we did for \texttt{HNSW}. 
Figure~\ref{fig:cost IVF} shows the results.
According to the results, on datasets with various dimensions from 256 to 960, DCOs take from 85.0\% to 95.3\% of the total running time of \texttt{IVF}. 
%

\smallskip
\noindent\textbf{Positive v.s. Negative Objects.} We verify empirically that for \texttt{IVF}, the number of DCOs on negative objects is significantly larger than that of DCOs on positive objects. 
The results are shown in Figure~\ref{fig:cost IVF}. 
According to the results, the number of negative objects is 60x to 869x more than that of the positive ones.

\subsubsection{Other AKNN Algorithms}

{\JIANYANG
In other AKNN algorithms {\CHENGB including tree-based, hashing-based, and quantization-based methods}, DCOs are also ubiquitous. 
{\JIANYANGREVISION Tree-based methods~\cite{muja2014scalable, dasgupta2008random, ram2019revisiting, beygelzimer2006cover, reviewer_M_tree}} generate candidate vectors through tree routing and {\JIANYANGLAST find out} KNNs with DCOs {\JIANYANGB (similarly as \texttt{IVF} does).}
Hashing-based methods~\cite{indyk1998approximate, datar2004locality, c2lsh, tao2010efficient, huang2015query, sun2014srs, lu2020vhp, zheng2020pm} generate candidate vectors via hashing {\CHENGC codes} and {\JIANYANGLAST find out} KNNs with DCOs (similarly as \texttt{IVF} does). 
Quantization-based methods~\cite{jegou2010product, learningtohash,ge2013optimized, ITQ, imi, additivePQ} generate candidates with short quantization codes, and conduct re-ranking (for {\JIANYANGLAST finding out} KNNs) with DCOs (similarly as \texttt{IVF} does).
%
{\CHENG For tree-based and hashing-based methods,}
the cost of DCOs is dominant because (1) one time tree routing or hashing bucket probing generates multiple candidates ({\CHENG which entail} multiple DCOs) and (2) tree routing and hashing bucket probing are much faster than a DCO (which has the time complexity of $O(D)$). 
{\CHENG For product quantization-based methods~\cite{jegou2010product, ge2013optimized, ITQ, imi}, DCOs are involved in its re-ranking stage, which is less dominant because the main cost lies in evaluating quantization codes.} {\JIANYANGREVISION {\chengr For comparison, we show the time decomposition results of \texttt{IVFPQ}, which is a quantization-based method~\cite{jegou2010product},}
in Figure~\ref{fig:cost IVFPQ} under the typical setting of \cite{learningtohash, jegou2010product}.}
}



\section{The \texttt{ADSampling} Method}
\label{sec:adsampling}

{\CHENGB Recall that our goal is to achieve \emph{reliable} DCOs with better efficiency than \texttt{FDScanning}. To this end,}
we develop a new method called \texttt{ADSampling}. 
At its core, \texttt{ADSampling} projects the objects to vectors with \emph{fewer} dimensions and conduct DCOs based on the projected vectors for better efficiency.
Different from the conventional and widely-adopted random projection technique~\cite{johnson1984extensions, datar2004locality, sun2014srs, c2lsh}, which projects \emph{all} objects to vectors with \emph{equal} dimensions, \texttt{ADSampling} is novel in the following aspects. 
First, it projects \emph{different} objects to vectors with \emph{different} numbers of dimensions during the query phase \emph{flexibly}.
We will elaborate on details of how this idea is implemented in Section~\ref{subsec:idea-1}. 
Second, it decides the number of dimensions to be sampled for each object \emph{adaptively} based on 
{\CHENGC the DCO on the object}
during the query phase, but not pre-sets it to a certain number during the index phase (which is knowledge demanding and difficult to set in practice).
We will elaborate on details of how this idea is implemented in Section~\ref{subsec:idea2}. 
In addition, we summarize \texttt{ADSampling} and 
{\JIANYANG prove}
that it has its time \emph{logarithmic} wrt $D$ for negative objects (which is significantly better than the time complexity $O(D)$ of \texttt{FDScanning}) in Section~\ref{subsection:theoretical analysis of ADSampling}.
\subsection{\textbf{Dimension Sampling over {\JIANYANG Randomly} Transformed Vectors}}
\label{subsec:idea-1}

{\CHENG 
{\CHENGB For better efficiency of a DCO, a natural idea is} to conduct a random projection~\cite{johnson1984extensions, vershynin_2018} on an object (i.e., to multiply the object (specifically its vector) with a $\mathbb{R}^{d\times D} $ random matrix $P$ where $d < D$~\footnote{{\JIANYANG There are multiple types of random matrices used for random projection~\cite{blockjlt, fftjlt, datar2004locality, johnson1984extensions}. 
In the present work, by random projection, we refer to the random projection based on random orthogonal matrix, which can be generated through orthonormalizing a random Gaussian matrix, {\JIANYANGB whose entries are independent standard Gaussian random variables}~\cite{choromanski2017unreasonable, randomortho, johnson1984extensions, vershynin_2018}.}}),
{\CHENGB and then conduct the DCO using the approximate distance that can be computed based on the projected vector, namely $\sqrt {D/d } \| P \mathbf{x}\|$.}
It is well-known that there exists a concentration inequality on the approximate distance as presented in the following lemma~\cite{vershynin_2018}.
\begin{lemma}
    For a given object $\mathbf{x} \in \mathbb{R}^D  $, a random projection $P \in \mathbb{R}^{d\times D} $ preserves its Euclidean norm with $\epsilon $ multiplicative error bound with the probability of
\begin{align}
    \mathbb{P} \left\{ \left| \sqrt {\frac{D}{d} } \| P \mathbf{x}\| -\| \mathbf{x} \| \right| \le \epsilon \| \mathbf{x} \|  \right\}   \ge 1 - 2e^{-c_0 d \epsilon ^2} 
    \label{equ:concentration}
\end{align}
    where $c_0$ is a constant factor and {\JIANYANGREVISION $\epsilon \in (0, +\infty)$}. \label{eq:fail_concen}\label{eq:concen}
    \label{lemma:concen}
\end{lemma}
Nevertheless, once an object is projected, the corresponding approximate distance would have a certain resolution {\CHENGC that would be fixed}.
Therefore, it lacks of flexibility of achieving different reduced dimensionalities for different objects (correspondingly different resolutions of approximate distances) during the query phase.

{\CHENGB We aim to project \emph{different} objects to vectors with \emph{different} numbers of dimensions during the query phase \emph{flexibly}.}
To this end, we propose to \emph{randomly transform} an object (with \emph{random orthogonal transformation}~\cite{johnson1984extensions, vershynin_2018, ITQ}, geometrically, to randomly rotate it) and then flexibly sample dimensions of the transformed vector for computing an approximate distance.
Formally, given an object $\mathbf{x}$, we first apply a \emph{random orthogonal matrix} $P' \in \mathbb{R}^{D \times D} $
to $\mathbf{x}$ and then sample $d$ rows on it (for simplicity, the first $d$ rows).
The result is denoted by $(P' \mathbf{x})|_{[1,2,...,d]}$. 
%
%
{\CHENGB This method entails two benefits.} 
First, we achieve the flexibility since we can sample $d$ dimensions of a rotated vector for different $d$'s during the query phase. Second, we achieve a guaranteed error bound since sampling $d$ dimensions on a transformed vector is equivalent to obtaining a $d$-dimensional vector via random projection, which we explain as follows. 

Recall that a random projection on $\mathbf{x}$ is to apply a random projection matrix $P \in \mathbb{R}^{d\times D}$ to $\mathbf{x}$, and the result is denoted by $P \mathbf{x} $.}
We claim that $(P' \mathbf{x})|_{[1,2,...,d]}$ (the result of our proposed method) and $P \mathbf{x}$ (the result of a random projection) are identically distributed. This is based on an elementary property of matrix multiplication 
that row samplings before and after a matrix multiplication are identical:
\begin{align}
    (P'\mathbf{x} )|_{[1,2,...,d]} = P'|_{[1,2,...,d]} \mathbf{x} 
\end{align}
%
We note that {\CHENGB $P'|_{[1,2,...,d]}$ corresponds to a random matrix $P$ for random projection since} one conventional way to generate a random projection matrix $P$ is to sample rows of a $D\times D$ random orthogonal matrix~\cite{choromanski2017unreasonable}.
%
%
Therefore, the concentration inequality for random projection over raw objects (as given in Equation~(\ref{equ:concentration})) can be applied to dimension sampling over {\JIANYANG randomly} transformed vectors, which provides solid foundation for our following discussion. 
%

{\CHENGB We denote the transformed vector as $\mathbf{y} := P' \mathbf{x}$.}
{\CHENG Based on the sampled dimensions, we can compute an approximate distance of $\mathbf{x}$, denoted by $dis'$, as follows,
\begin{align}
    dis' := \sqrt { \frac{D}{d}} \left\| \mathbf{y}|_{[1,2,...,d]} \right\|  
\end{align}
where $d$ is the number of sampled dimensions. We note that the time complexity of computing an approximate distance based on $d$ sampled dimensions is $O(d)$.
Furthermore, when all $D$ dimensions are sampled, the distance $dis'$ computed based on the sampled dimensions would be equal to the true distance $dis$, which is due to the fact
{\JIANYANG that random orthogonal transformation preserves the norm of any vector (since it simply rotates the space without distorting the distances). 
}
}
\subsection{{\JIANYANGREVISION \textbf{Incremental Sampling with Hypothesis Testing}}}
\label{subsec:idea-2}
\label{subsec:idea2}

{\CHENG 
One remaining issue is how to determine the number of dimensions of $\mathbf{y}$ we need to sample in order to make a sufficiently confident conclusion for the DCO (i.e., to decide whether $dis \le r$). Intuitively, with more sampled dimensions, the approximate distance $dis'$ would be 
{\JIANYANGB more accurate, }
and we would be able to make a more confident conclusion. On the other hand, sampling more dimensions would result in higher cost of computing the approximate distance (since the cost is linear wrt the number of sampled dimensions). We aim to sample the minimum possible number of dimensions, which are sufficient to make a confident conclusion. 

Specifically, we propose to sample the dimensions of $\mathbf{y}$ in an 
{\JIANYANGREVISION \emph{incremental}}
manner, i.e., we start with a few dimensions. If with the current sampled dimensions, we cannot make a confident conclusion, we continue to sample 
{\JIANYANGLAST some more}
until we can make a confident conclusion or we have sampled all dimensions. 
%
As a result, the problem reduces to the one of deciding whether we can make a sufficiently confident conclusion with a certain, say $d$, sampled dimensions? In a statistics language, the observed distance $dis'$ (computed based on the sampled dimensions) is an estimator of the true distance $dis$ and its distribution depends only on the true value $dis$ and the number of sampled dimensions $d$. The task is to draw a conclusion about a true value $dis$ (i.e., whether $dis \le r$) with an observed value $dis'$. It's exactly what \emph{hypothesis testing} typically does. Motivated by this, we propose to leverage hypothesis testing to solve the problem.
Specifically, we conduct the hypothesis testing as follows.

\begin{enumerate}
    \item We define a null hypothesis $H_0:dis \le r$ and its alternative $H_1: dis > r$.
    \item We use $dis'$ as the estimator of $dis$. The relationship between $dis'$ and $dis$ is provided in Lemma~\ref{eq:concen} (i.e., the difference between $dis'$ and $dis$ is bounded by $\epsilon\cdot dis$ with the failure probability at most $2\exp (-c_0 d  \epsilon^2)$).
    \item We set the significance level $p$ to be $2 \exp(-c_0  \epsilon_0^2)$, where $\epsilon_0$ is a parameter to be tuned empirically. 
    {\CHENGB With this,} the event that the observed $dis'$ is much larger than $r$ (i.e., $dis' > (1 + \epsilon_0/\sqrt{d})\cdot r$) has its probability below the significance level $p$ (which can be verified based on Lemma~\ref{eq:concen} with $\epsilon = \epsilon_0/\sqrt{d}$ and $H_0:dis \le r$).
    \item We check whether 
    {\CHENGB the event happens} ($dis' > (1 + \epsilon_0/\sqrt{d})\cdot r$). If so, {\CHENGB we can reject $H_0$ and conclude} $H_1: dis > r$ {\CHENGB with sufficient confidence}; otherwise, we cannot.
\end{enumerate}

There are three cases for the outcome of the hypothesis testing. \underline{Case 1}: we reject the hypothesis (i.e., we conclude $dis > r$) and $d < D$. In this case, the time cost (which is mainly for evaluating the approximate distance) is $O(d)$, which is smaller than that of computing the true distance in $O(D)$ time. \underline{Case 2}: we cannot reject the hypothesis and $d < D$. In this case, we would continue to sample some more dimensions of $\mathbf{y}$ \emph{incrementally} and conduct another hypothesis testing. \underline{Case 3}: $d = D$. In this case, we have sampled all dimensions of $\mathbf{y}$ and the approximate distance based on the sampled vector is equal to the true distance. Therefore, we can conduct an \emph{exact} DCO. We note that the 
{\JIANYANGREVISION incremental} dimension sampling process with (potentially sequential) hypothesis testing would have its time cost strictly smaller than $O(D)$ (when it terminates in Case 1) and equal to $O(D)$ (when it terminates in Case 3). 
}

{\CHENGB We note that hypothesis testing has also been used for deciding a certain number of hashes for LSH in the context of similarity search~\cite{sequentialLSH, satuluri2011bayesian}. The differences between our technique and \cite{sequentialLSH, satuluri2011bayesian} include: (1) ours is based on a random process of sampling dimensions of a transformed vector while \cite{sequentialLSH, satuluri2011bayesian} are on one of sampling hash functions, which entail significantly different hypothesis testings and (2) ours targets the Euclidean distance function while \cite{sequentialLSH, satuluri2011bayesian} target similarity functions such as Jaccard and Cosine similarity measures (it remains non-trivial to adapt the latter to the Euclidean space),}
{\JIANYANGB and (3) ours guarantees to be no worse than the method of evaluating exact distances (in our case, i.e., \texttt{FDScanning}) because it obtains exact distances when it has sampled all the dimensions while \cite{sequentialLSH, satuluri2011bayesian} have no such guarantee (when they have sampled all the hash functions and still cannot produce a firmed result, they would have to re-evaluate exact similarities from scratch).}

\subsection{Summary and Theoretical Analysis}
\label{subsection:theoretical analysis of ADSampling}

{\CHENG \noindent\textbf{Summary.} 
We summarize the process of \texttt{ADSampling} in Algorithm~\ref{code:adasampling}.
{\JIANYANGREVISION It takes a transformed data vector $\mathbf{o}'$, a transformed query vector $\mathbf {q}'$ and a distance threshold $r$ as inputs and outputs the result of the DCO of whether $dis \le r$: 1 for yes {\CHENGC (in this case, it returns $dis$ as well)} and 0 for no. 
We note that the transformation of the data vectors is conducted in the index phase and its cost can be amortized by all the subsequent queries on the same database. The transformation of the query vector is conducted in the query phase when a query comes and its cost can be amortized by all the DCOs involved for answering the same query.} 
Specifically, the algorithm maintains the number of sampled dimensions with a variable $d$ with $d = 0$ initially (line 1). 
It then performs an iterative process if $d < D$ (line 2). 
At each iteration, it samples some more dimensions incrementally and updates $d$ {\JIANYANGB and the approximate distance $dis'$} accordingly (line 3-4)
and conducts a hypothesis testing with the null hypothesis as $dis \le r$ based on the approximate distance $dis'$ (line 5). 
It then returns the result in three cases as explained in Section~\ref{subsec:idea-2} (line 6 - 11). 

\begin{algorithm}[tbh]
\DontPrintSemicolon
\SetKwFunction{LevelUp}{LevelUp}
\SetKwData{and}{and}
\SetKwInOut{Input}{Input}\SetKwInOut{Output}{Output}
\Input{{\JIANYANGREVISION A transformed data vector $\mathbf{o}'$, a transformed query vector $\mathbf q'$ and a distance threshold $r$}}
\Output{The result of DCO (i.e., whether $dis \le r$): 1 means yes and 0 means no; {\CHENG In case of the result of 1, an {\CHENGB exact} distance is also returned}}
\BlankLine
{\CHENGB Initialize the number of sampled dimensions $d$ to be 0\;}
\While{$d < D$}{
        {\JIANYANGREVISION Sample some more dimensions $y_i$ {\chengr incrementally with} $y_i=\mathbf{o}_i'-\mathbf{q}_i'$\;}
        Update $d$ and the approximate distance $dis'$ accordingly\;
	Conduct a hypothesis testing with the null hypothesis {\JIANYANGB $H_0$} as $dis \le r$ based on the approximate distance $dis'$\;
	\If(\tcp*[f]{Case 1}){{\JIANYANGB $H_0$} is rejected and $d < D$}{
			\textbf{return} 0  \;
		}
	\ElseIf(\tcp*[f]{Case 2}){{\JIANYANGB $H_0$} is not rejected and $d < D$}{
		\textbf{continue}\;
		}
	\Else(\tcp*[f]{Case 3}){
	    \textbf{return} 1 {\CHENG (and $dis'$)} if $dis' \le r$ and 0 otherwise\;
	}
}
\caption{\texttt{ADSampling}}
\label{code:adasampling}
\end{algorithm}

\vspace{-4mm}
\smallskip\noindent\textbf{Failure Probability Analysis.}
Note that \texttt{ADSampling} terminates in either Case 1 (with the hypothesis being rejected and $d < D$) or Case 3 (with $d = D$). When it terminates in Case 3, there would be no failure since in this case, the approximate distance $dis'$ is equal to the true distance $dis$ and the DCO result is exact. When it terminates in Case 1, a failure would happen if $dis \le r$ holds since in this case, it concludes that $dis > r$ (by rejecting the null hypothesis). We analyze the probability of the failure. {\JIANYANG As discussed in Section~\ref{subsec:idea-2}, we can control the failure probability with $\epsilon_0$. The following lemma presents the relationship between $\epsilon_0$ and the failure probability of a DCO with \texttt{ADSampling}. } 
}

\begin{lemma}
For a DCO in $D$-dimensional space, the failure probability of \texttt{ADSampling} is given by
\begin{align}
    &\mathbb{P}\left\{ {\CHENG failure}  \right\}  =0 {\CHENG \text{~~if $dis > r$}}
    \\&
    \mathbb{P}\left\{ {\CHENG failure} \right\}  \le \exp \left( -c_0 \epsilon_0^2 + \log D \right) 
    {\CHENG \text{~~if $dis \le r$}}
\end{align}
\label{theorem:ADSampling accuracy}
\label{lemma:ADSampling accuracy}
\end{lemma}
\begin{proof}
{\CHENG The correctness for the case of $dis > r$ is obvious and that of the other case ($dis \le r$) can be verified as follows.}
\begin{align}
    &\mathbb{P}\left\{ {\CHENG failure} \right\} = \mathbb{P} \left\{ \exists d < D, dis' > (1 + \epsilon_0 / \sqrt {d} ) \cdot r \right\} \label{eq:rejection}
    \\\le &\sum_{d=1}^{D-1} \mathbb{P} \left\{ dis' > (1 + \epsilon_0 / \sqrt {d} ) \cdot dis \right\}  \label{eq:union bound}
    \\\le &\sum_{d=1}^{D-1} \exp \left( -c_0 \epsilon_0^{2}  \right)\le \exp \left( -c_0 \epsilon_0^2 + \log D \right) \label{eq:concentrationlemma3.1}
\end{align}
{\JIANYANG 
where (\ref{eq:rejection}) is because a failure happens if and only if we accidentally reject the hypothesis for some $d<D$; 
(\ref{eq:union bound}) applies union bound and the fact that $dis \le r$; and (\ref{eq:concentrationlemma3.1}) is due to Lemma~\ref{eq:concen}.
}
\end{proof}




{\CHENG

\smallskip\noindent\textbf{Time Complexity Analysis.}
Let $\hat{D}$ be the number of sampled dimensions by \texttt{ADSampling}. Clearly, the time complexity \texttt{ADSampling} is $O(\hat{D})$. Given the stochastic nature of the method, $\hat{D}$ is a random variable. Next, we analyze the expectation of $\hat{D}$, denoted by $\mathbb{E}[\hat{D}]$. 
First of all, since $\hat{D}$ is always at most $D$, we know $\mathbb{E}[\hat{D}] \le D$. Furthermore, for the DCO on a negative object with $dis > r$, we can derive that $\mathbb{E}[\hat{D}]$ relies on $\epsilon_0$ and $\alpha = (dis-r)/r$ {\JIANYANG (which we call the \emph{distance gap} between $dis$ and $r$)}, as presented below (detailed proof can be found in 
{\JIANYANGREVISION Appendix~\ref{section:theory}).}
%
\begin{lemma}
When \texttt{ADSampling} is used for the DCO on an object and a threshold $r$ with $dis > r$, letting $\alpha =(dis-r)/r$, we have
\begin{align}
    \mathbb{E} \left[ \hat D  \right]  = O \left[ \min \left( D, \alpha _{}^{-2} \cdot \epsilon _{0}^{2}  \right)  \right] 
\end{align}
\label{theorem:ADSampling efficiency}
\end{lemma}
The above result is well aligned with the intuitions that (1) when the distance gap between $dis$ and $r$, i.e., $(dis-r)/r$, is larger, fewer dimensions would be sampled for making a sufficiently confident conclusion and (2) when $\epsilon_0$ is larger (i.e., the significance value of the hypothesis testings is smaller, which means a higher requirement on the confidence), more dimensions would be sampled.

We further derive the time-accuracy trade-off of \texttt{ADSampling}.
\begin{theorem}
When \texttt{ADSampling} is used for the DCO on an object and a threshold $r$ with $dis > r$, letting $\alpha =(dis-r)/r$, we have
\begin{align}
    \mathbb{E} \left[ \hat D \right]  = O \left[ \min \left( D, \frac{1}{\alpha ^2} \log \frac{D}{\delta}  \right)  \right] 
\end{align}
for achieving its failure probability {\JIANYANG (of positive objects)} at most $\delta$.
\label{theorem:time-accuracy of ADSampling}
\end{theorem}
}

{\JIANYANG 
\begin{proof}
Making the failure probability in Lemma~\ref{theorem:ADSampling accuracy} be equal to $\delta$, we obtain the corresponding $\epsilon_0$. Then by substituting $\epsilon_0$ in Lemma~\ref{theorem:ADSampling efficiency}, we have the theorem. 
\end{proof}

\smallskip
\noindent\textbf{\texttt{ADSampling v.s. FDScanning.}}
Compared with \texttt{FDScanning}, \texttt{ADSampling} improves the complexity for negative objects from being linear to being logarithmic wrt $D$ at the cost of the accuracy for positive objects (Theorem~\ref{theorem:time-accuracy of ADSampling}).
%
We emphasize that the failure probability (of positive objects) decays \textbf{quadratic-exponentially} (Lemma~\ref{theorem:ADSampling accuracy}) while the time complexity (of negative objects) grows \textbf{quadratically} (Lemma~\ref{theorem:ADSampling efficiency}), both with respect to $\epsilon_0$. It indicates that to achieve \emph{nearly-exact} DCOs, we only need sample a few dimensions.  
{\JIANYANG 
We empirically verify these results in Section~\ref{subsubsec:theoretical results}. It shows that with \texttt{ADSampling} as a plugin, an exact KNN algorithm, namely linear scan, needs only on average 55 dimensions per vector on GIST (originally 960 dimensions) to achieve >99.9\% recall. }

}

{\CHENG
\section{\texttt{AKNN+}: Improving AKNN Algorithms with \texttt{ADSampling} as a Plug-in Component}
\label{sec:aknn+}

\label{subsec:adsampling-any-aknn}
Recall that an AKNN algorithm, which we denote by \texttt{AKNN} and could be any one among many existing algorithms~\cite{malkov2018efficient, jegou2010product, muja2014scalable, fu2019fast, datar2004locality}, involves many DCOs. 
In the literature, \texttt{FDScanning} is typically adopted for DCOs and runs in $O(D)$ time.
Given that \texttt{ADSampling} can conduct 
{\CHENGB reliable DCOs with better efficiency,}
a natural idea is to improve the AKNN algorithms by adopting \texttt{ADSampling} for the DCOs.
{\JIANYANG Specifically,} since \texttt{ADSampling} is based on randomly transformed data vectors and query vectors, before any query comes, we randomly transform all data vectors, and when a query comes, we randomly transform the query vector. Then, we run the AKNN algorithm based on the transformed data and query vectors. Recall that the time cost of transforming the data vectors can be amortized across different queries and the time cost of transforming the query vector can be amortized across many different DCOs involved for answering the query. During the running process of the AKNN algorithm, whenever it conducts a DCO, we use the \texttt{ADSampling} method. For example, for graph-based methods such as \texttt{HNSW}, we use the \texttt{ADSampling} method when comparing {\CHENGC the distance of} a newly visited object with the maximum in the result set $\mathcal R$. 
{\chengf For other AKNN algorithms such as} \texttt{IVF}, we apply \texttt{ADSampling} when {\CHENGB comparing the distance of a candidate and the maximum in the currently maintained KNN set $\mathcal K$} for selecting the final KNNs from the generated candidates.

For an AKNN algorithm \texttt{AKNN}, which adopts \texttt{ADSampling} for DCOs, we call it \texttt{AKNN+}. For example, we call \texttt{HNSW} and \texttt{IVF} with \texttt{ADSampling} adopted for DCOs \texttt{HNSW+} and \texttt{IVF+}, respectively. 

{\JIANYANG
\smallskip\noindent\textbf{{\CHENGB Theoretical Analysis}.} 
Recall that \texttt{ADSampling} improves the efficiency of DCOs on negative objects at the cost of the accuracy of those on positive objects. 
{\CHENGB We 
show the relationship between the probability that \texttt{AKNN+} fails to return the same results as \texttt{AKNN} and the time complexity of the DCO on a negative object involved in \texttt{AKNN+} below.} 
{\JIANYANGB Basically, to preserve the returned results of \texttt{AKNN}, it suffices to produce correct results for all DCOs, whose number is at most $N$. Then with union bound, the failure probability of \texttt{AKNN+} is upper {\CHENG bounded} by the sum of the failure probability of each single DCO. Thus, making the failure probability of \texttt{ADSampling} be $\delta = \delta' / N$ yields the following corollary.}
\begin{corollary}
Let $\delta'$ be the probability that \texttt{AKNN+} fails to return the same results as \texttt{AKNN}. The expected time complexity of the DCO on a negative object with distance gap $\alpha$ is reduced to
\begin{align}
    \mathbb{E} \left[ \hat D \right]  = O \left[ \min \left( D, \frac{1}{\alpha ^2} \log \frac{DN}{\delta'}  \right)  \right] 
\end{align}
and the remaining time cost ({\CHENGC for} DCOs on positive objects and other {\CHENGC computations}) is unchanged.
\end{corollary}

Furthermore, for those \texttt{AKNN+} algorithms which generate candidates all at once (e.g., {\CHENGC{\texttt{IVF}}}), 
producing correct DCO results for KNN objects ($K$ objects) rather than for all objects (at most $N$ objects) is sufficient to return the same results as its corresponding \texttt{AKNN} algorithm. This is because once we produce correct results for {\CHENGB the DCOs on} the true KNN objects, we also obtain their exact distances {\CHENGC (note that all of them would be positive objects)}. It ensures to return them as the final answers. Thus, we have the following corollary.

\begin{corollary}
Let $\delta'$ be the probability that \texttt{AKNN+} fails to return the KNNs of the generated candidates. The expected time complexity of DCO on a negative object with distance gap $\alpha$ is reduced to
\begin{align}
    \mathbb{E} \left[ \hat D \right]  = O \left[ \min \left( D, \frac{1}{\alpha ^2} \log \frac{DK}{\delta'}  \right)  \right] 
\end{align}
and the remaining time cost ({\CHENGC for} DCOs on positive objects and other {\CHENGC computations}) is unchanged.
\label{corollary: find out KNNs}
\end{corollary}

}
}

{\CHENG
\section{\texttt{AKNN++}: Improving \texttt{AKNN+} Algorithms with Algorithm Specific Optimizations}
\label{sec:aknn++}

\subsection{\texttt{HNSW++}: Towards More Approximation}
\label{subsec:hnsw++}

Recall that \texttt{HNSW+} maintains a result set $\mathcal{R}$ with a max-heap of size $N_{ef}$ and distances as keys, where $N_{ef} > K$. For each newly generated candidate object, it checks whether its distance is 
{\JIANYANGLAST no greater than}
the largest distance (of an object) in $\mathcal{R}$ and if so, it inserts the object in the set $\mathcal{R}$. Specifically, it uses \texttt{ADSampling} to conduct the DCO for each candidate object with the largest distance in $\mathcal{R}$ as the threshold distance. We identify two roles played by the set $\mathcal{R}$. 
\underline{First}, it maintains the KNNs with the smallest distances among those candidates generated so far. These KNNs would be returned as the outputs of the algorithm at the end. 
\underline{Second}, it maintains the $N_{ef}^{th}$ largest distance among the candidates generated so far. This distance is used as the threshold distance of the DCOs through the course of the algorithm, whose results would affect how the candidates are generated. 
{\JIANYANG Specifically, if a candidate {\CHENGC generated by} \texttt{HNSW+} has its distance at most the $N_{ef}^{th}$ distance in $\mathcal{R}$, it would be {\CHENGB added to} the search set $\mathcal S$ for further candidate generation.}

This dual-role design is attributed to the fact that in \texttt{HNSW+}, \emph{exact} distances are used for fulfilling both roles. 
{\JIANYANG As shown in 
Figure~\ref{fig:illustration HNSW+}, \texttt{HNSW+} always maintains $\mathcal{R} $ and $\mathcal{S} $ with exact distances (dark green), and the first $K$ objects in $\mathcal{R}$ are the KNN objects.}
Using the exact distances is desirable for the first role (of maintaining the KNNs) 
since the outputs of the algorithms are defined based on the exact distances. Yet we argue that it may not be cost-effective for the second role (of maintaining the $N_{ef}^{th}$ largest distance) since the procedure that uses this distance for generating candidates is a heuristic one (i.e., greedy beam search) and may still work well with an approximate distance. 

Therefore, we propose to decouple the two roles of $\mathcal{R}$ by maintaining two sets $\mathcal{R}_1$ and $\mathcal{R}_2$, one for each role {\JIANYANG (as illustrated in 
Figure~\ref{fig:illustration HNSW++})}. 
Set $\mathcal{R}_1$ has a size of $K$ and is based on exact distances {\JIANYANG (dark green)}. Set $\mathcal{R}_2$ has its size of $N_{ef}$ and is based on distances, {\CHENGC which could be either exact or approximate}. Specifically, for each newly generated candidate, it checks whether its distance is 
{\JIANYANGLAST no greater than}
the maximum distance in set $\mathcal{R}_1$, and if so, it inserts the candidate in set $\mathcal{R}_1$. 
{\JIANYANG Furthermore, this DCO produces a by-product, namely the observed distance $dis'$ (light green) when {\CHENG \texttt{ADSampling} terminates}, which could be exact (if all $D$ dimensions are sampled) or approximate (if it terminates with $d<D$). }
It then maintains the set $\mathcal{R}_2$ and the set $\mathcal{S}$ based on the 
{\JIANYANG observed} distances similarly as \texttt{HNSW+} maintains $\mathcal{R}$ and $\mathcal{S}$, respectively. We call the resulting algorithm that is based on this decoupled-role design \texttt{HNSW++}. 

\begin{figure}[thb]
    \centering
    \vspace{-4mm}
    \captionsetup[subfigure]{aboveskip=-1pt}
    \begin{subfigure}[b]{0.45\linewidth}
        \centering
        \includegraphics[width=0.8\textwidth]{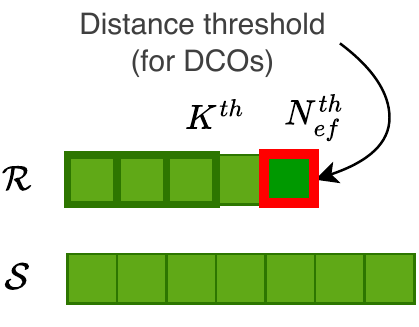}
        \caption{\texttt{HNSW+}}
        \label{fig:illustration HNSW+}
    \end{subfigure}   
    \begin{subfigure}[b]{0.45\linewidth}
        \centering
        \includegraphics[width=0.8\textwidth]{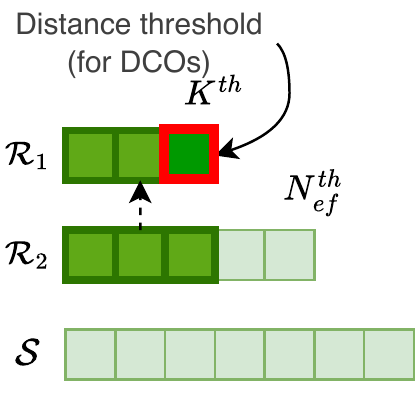}
        \caption{\texttt{HNSW++}}
        \label{fig:illustration HNSW++}
    \end{subfigure}   
    \vspace{-4mm}
    \caption{\texttt{HNSW+} v.s. \texttt{HNSW++}}
    \vspace{-4mm}
    \label{fig:illustration HNSW+ v.s. HNSW++}
\end{figure}
\smallskip\noindent\textbf{{\CHENGB Theoretical Analysis}.}
We note that different from \texttt{HNSW+}, which would return the same results as \texttt{HNSW} with high probability, \texttt{HNSW++} does not aim to return the same results as \texttt{HNSW} {\JIANYANG (though in practice, it returns nearly the same results as verified in Section~\ref{subsub:dimensions and recall}). 
{\CHENG Specifically, \texttt{HNSW++} would generate a set of candidates, which might be different from that of \texttt{HNSW+} or \texttt{HNSW}. 
{
Among the generated candidates, \texttt{HNSW++} guarantees to return their KNNs with high probability because it still maintains KNNs with \texttt{ADSampling}, {\CHENG and its guarantee is the same as the one in Corollary~\ref{corollary: find out KNNs}}. 
}
}
}

\smallskip\noindent\textbf{\texttt{HNSW++} v.s. \texttt{HNSW+}.}
Compared with \texttt{HNSW+}, \texttt{HNSW++} is expected to have a better time-accuracy trade-off, which we explain as follows. \underline{First}, consider the time cost. In \texttt{HNSW++}, for each DCO, the threshold distance is the $K^{th}$ largest distance, which is smaller than that used in \texttt{HNSW+} (i.e., the $N_{ef}^{th}$ largest distance). Correspondingly, in \texttt{HNSW++}, the $\alpha$ value, which is defined as $(dis - r)/r$, is larger than that in \texttt{HNSW+}. Therefore, the time cost for this DCO would be smaller than that in \texttt{HNSW+} according to the time complexity analysis of \texttt{ADSampling} in Section~\ref{subsection:theoretical analysis of ADSampling}. \underline{Second}, consider the effectiveness. While \texttt{HNSW++} and \texttt{HNSW+} use different distances for generating the candidates, we expect that they would generate candidates with similar qualities given that (1) the distances used by the two algorithms should be close (or the same in some cases) and (2) the method used for generating candidates, i.e., greedy beam-search, has a heuristic nature and there is no strong clue that it favors exact distances over approximate ones.

{ \JIANYANG
\noindent\textbf{Remarks.} We note that the technique of \texttt{HNSW++} can also be used in other 
graph-based methods~\cite{malkov2018efficient, li2019approximate, fu2019fast, diskann, NSW}.
This is because these algorithms also apply the greedy beam search {\CHENGC based on a set $\mathcal{R}$} in the query phase.
}

\subsection{\texttt{IVF++}: Towards Cache Friendliness}
\label{subsec:ivf++}
}

{\JIANYANG




In the original \texttt{IVF} algorithm, the vectors in the same cluster are stored sequentially. When evaluating their distances, the algorithm scans all the dimensions of these vectors \emph{sequentially}, which exhibits strong locality of reference, {\CHENG and thus it is} cache-friendly. {\CHENG Figure~\ref{fig:data layout plain} illustrates the corresponding data layout {\JIANYANGB (as indicated by the arrow)} and 
{\CHENGB the data needed (as indicated by the colored background).}
} In {\CHENG \texttt{IVF+}}, 
though {\CHENG it scans} fewer dimensions than {\texttt{IVF}}, 
it would not be cache-friendly with the same data layout. 
Specifically, when {\CHENG \texttt{IVF+}} terminates {\CHENG the dimension sampling process} for a data vector, 
the subsequent dimensions {\CHENG would probably have been} loaded into cache from main memory though they are not needed. {\CHENG Figure~\ref{fig:data layout ARSearch} illustrates the corresponding data layout {\CHENGB and data needed.}
}
}

{\CHENG
We propose to re-organize the data layout of the candidates and adjust the order of the dimensions of the candidates to be fed to \texttt{ADSampling} accordingly so as to achieve more cache-friendly data accesses. Recall that for each candidate, \texttt{ADSampling} would definitely sample a few, say $d_1$, dimensions of the candidate first and then 
{\JIANYANGREVISION incrementally}
sample more dimensions depending on the hypothesis testing outcomes. That is, the first $d_1$ dimensions of each candidate would be accessed for sure. Motivated by this, we store the first $d_1$ dimensions of all candidates sequentially in an array $A_1$ and the remaining $D - d_1$ dimensions of all candidates sequentially in another array $A_2$. 
We note that the process of re-organizing the data layout can be conducted during the index phase. During the query phase, when using \texttt{ADSampling} for DCOs on the candidates, we follow the following order of the dimensions of the candidates: the first $d_1$ dimensions of the first candidate, the first $d_1$ dimensions of the second candidate, ..., the first $d_1$ dimensions of the last candidate, the $D-d_1$ dimensions of the first candidate, ..., the $D-d_1$ dimensions of the last candidate. 
Figure~\ref{fig:data layout ARScan} 
illustrates the corresponding data layout {\CHENGB and data needed.}
We call the resulting algorithm \texttt{IVF++}. \texttt{IVF++} and \texttt{IVF+} would produce exactly the same results, but the former is more cache friendly since it utilizes the locality of reference for the first $d_1$ dimensions of all candidates.
}

\begin{figure}[thb]
    \centering
    \vspace{-2mm}
    \begin{subfigure}[b]{0.3\linewidth}
        \centering
        \includegraphics[height=0.9\textwidth]{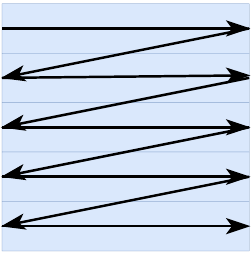}
        \caption{{\CHENG \texttt{IVF}}}
        \label{fig:data layout plain}
    \end{subfigure}       
    \begin{subfigure}[b]{0.3\linewidth}
        \centering
        \includegraphics[height=0.9\textwidth]{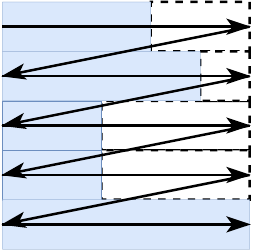}
        \caption{{\CHENG \texttt{IVF+}}}
        \label{fig:data layout ARSearch}
    \end{subfigure}       
    \begin{subfigure}[b]{0.3\linewidth}
        \centering
        \includegraphics[height=0.9\textwidth]{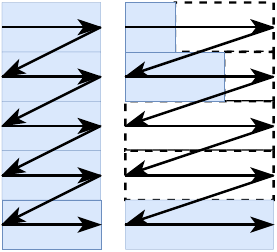}
        \caption{{\CHENG \texttt{IVF++}}}
        \label{fig:data layout ARScan}
    \end{subfigure}       
    \vspace{-4mm}
    \caption{Data Layout {\CHENGC and Data Needed}}
    \vspace{-4mm}
    \label{fig:data layout}
\end{figure}

{\CHENGB
\smallskip
\noindent\textbf{Theoretical Analysis.}}
{\JIANYANG 
Since \texttt{IVF++} and \texttt{IVF+} differ only in data layout, they have the same theoretical guarantee (Corollary~\ref{corollary: find out KNNs}).
}

{\CHENG
\smallskip
\noindent\textbf{Remarks.} We note that the technique used for improving \texttt{IVF+} with cache friendliness can also be used for improving some other \texttt{AKNN+} algorithms, including those of tree-based methods~\cite{muja2014scalable, dasgupta2008random, ram2019revisiting, beygelzimer2006cover}, quantization-based methods~\cite{jegou2010product, imi} and hashing-based methods~\cite{datar2004locality, c2lsh, Dong2020Learning}.
This is because all these algorithms generate the candidates in a batch and then re-rank the candidates for finding out KNNs.
}

\section{Experiment}

\captionsetup[subfigure]{aboveskip=10pt}

\label{section:experiment}


\subsection{Experimental Setup}
\label{subsec: experimental setup}
\noindent
\textbf{Datasets.} We {\CHENG use} six public datasets with varying {\CHENG sizes and dimensionalities}~\footnote{{\JIANYANG Note that our techniques introduce nearly no extra space consumption (the only extra space consumption is brought by a $D\times D$ random orthogonal matrix, which is ignorable compared with the huge $D$-dimensional database of size $N$). Thus, {\CHENGC they do not affect} the scalability of the AKNN algorithms. We thus focus on million-scale datasets to verify {\CHENGC their} effectiveness in speeding up the AKNN algorithms.}}, whose details are shown in Table~\ref{tab:data}. These datasets {\CHENG have been} widely used to benchmark AKNN algorithms ~\cite{lu2021hvs,adaptive2020ml, li2019approximate}.
{\JIANYANGREVISION We note that these  public datasets 
provide both data and query vectors}.

\begin{table}[h]
\caption{Dataset Statistics}
\vspace{-4mm}
\label{tab:data}
\begin{tabular}{c|cccc}
\hline
Dataset & Size      & $D$ & {\JIANYANGREVISION Query Size} & Data Type \\ \hline
Msong   & 992,272   & 420       & {\JIANYANGREVISION 200}     & Audio     \\
DEEP    & 1,000,000 & 256       & {\JIANYANGREVISION 1,000}     & Image     \\
Word2Vec & 1,000,000 & 300      & {\JIANYANGREVISION 1,000}      & Text      \\
GIST    & 1,000,000 & 960       & {\JIANYANGREVISION 1,000}     & Image     \\
GloVe   & 2,196,017 & 300       & {\JIANYANGREVISION 1,000}     & Text      \\
Tiny5M  & 5,000,000 & 384       & {\JIANYANGREVISION 1,000}     & Image     \\ \hline
\end{tabular}
\end{table}

\smallskip
\noindent
{\JIANYANGREVISION \textbf{{\CHENGB Algorithms}.} }
{\CHENGB For reliable DCOs, we compare our proposed method \texttt{ADSampling} with the conventional \texttt{FDScanning} {\JIANYANGREVISION and \texttt{PDScanning} (Partial Dimension Scanning), which we explain below. \texttt{PDScanning} 
incrementally scans the dimensions of {\chengr a} raw vector 
{\chengr and} terminates {\chengr the process} when the distance based on the partially scanned $d$ dimensions, i.e., $\sqrt {\sum_{i=1}^d x_i^2}$, is greater than the distance threshold $r$. {\chengr We note that \texttt{PDScanning} starts with zero dimensions but not a pre-set number of dimensions since (1) it is hard to set the number and (2) starting from a certain number of dimensions or zero dimensions have very similar performance given the fact that the dimensions are scanned incrementally.} We also note that \texttt{PDScanning} is an exact algorithm for {\chengr DCOs} and has the worst-case time complexity of $O(D)$. We name the AKNN algorithms {\chengr with} \texttt{PDScanning} {\chengr for DCOs} as 
\texttt{AKNN}*
and the one with a further optimized data layout as 
\texttt{AKNN}**
(for \texttt{IVF} only).} We exclude those distance approximation methods such as product quantization and random projection from comparison since as explained in Section~\ref{sec:introduction} and further verified in Section~\ref{subsubsec:reliable-dco}, they can hardly achieve reliable DCOs. 
{\JIANYANGREVISION For AKNN algorithms, we  mainly focus on \texttt{HNSW}~\cite{malkov2018efficient} and \texttt{IVF}~\cite{jegou2010product} for providing the contexts of DCOs since they correspond to two state-of-the-art AKNN algorithms as benchmarked in~\cite{annbenchmark, li2019approximate}. We note that these methods are widely adopted in industry (including 
Faiss~\cite{johnson2019billion}, 
Milvus~\cite{milvus} and PASE~\cite{PASE}).}}
{\JIANYANGREVISION For better comprehensiveness, we also consider one of the best tree-based methods \texttt{Annoy}~\cite{annoy} (as benchmarked in \cite{annbenchmark, li2019approximate}) and a hashing-based method \texttt{PMLSH}~\cite{zheng2020pm}. We note that their performance of time-accuracy tradeoff is suboptimal compared with \texttt{HNSW} and \texttt{IVF}. 
Due to the limit of space, 
we include their results in 
Appendix~\ref{appendix:section tree and hashing}.}

\smallskip
\noindent
{\JIANYANGREVISION \textbf{{\CHENGB Performance Metrics}.} }
{\JIANYANGREVISION We use two metrics to measure the accuracy: (1) 
recall~\cite{annbenchmark, li2019approximate, malkov2018efficient, jegou2010product}, i.e., the ratio between the number of successfully retrieved ground truth KNNs and $K$ and (2) {\chengr average distance} ratio~\cite{reviewer_paper, c2lsh, huang2015query, sun2014srs, reviewer_SISAP_metric}, i.e., the average of the {\chengr distance ratios 
(which equals to the average relative error on distance {\chengr plus one})
of}
the retrieved $K$ objects {\chengr wrt} the ground truth KNNs.
We adopt the query-per-second (QPS), i.e., the number of handled queries per second, to measure efficiency. }
{\JIANYANGB Note that the query time is measured \textit{end-to-end} (i.e., including the time of random transformation on query vectors). {\JIANYANGREVISION We decompose the time cost in Section~\ref{subsec: time decompostion}.}}
{\CHENG We also measure the total number of dimensions evaluated by an algorithm. For \texttt{AKNN} algorithms, it means the total number of dimensions of the candidates 
(since for each candidate, all of its dimensions are used for computing its distance). 
{\JIANYANGREVISION For \texttt{AKNN+} (\texttt{AKNN}*) and \texttt{AKNN++} (\texttt{AKNN}**) algorithms, it means the total number of {\chengr sampled (scanned)} dimensions of the candidates (since for a candidate, only those sampled {\chengr (scanned)} dimensions are used for computing its distance approximately). All the mentioned metrics are averaged over the whole query set.}
}

\smallskip
\noindent
\textbf{Implementation.} 
The implementation of an AKNN algorithm consists of two phases. During {\CHENG the} \underline{index phase}, we first generate a random orthogonal transformation matrix with the NumPy library, store it and apply the transformation to all {\CHENG data} vectors. 
{\JIANYANGREVISION Then we feed the transformed vectors (the raw vectors for \texttt{AKNN}, \texttt{AKNN}* and \texttt{AKNN}**) into existing AKNN algorithms. In particular, 
for \texttt{HNSW}, \texttt{HNSW+},  \texttt{HNSW++} and \texttt{HNSW}*
(note that they have the same graph structure), our implementation is based on hnswlib~\cite{malkov2018efficient}, while 
for \texttt{IVF}, \texttt{IVF+}, \texttt{IVF++}, \texttt{IVF}* and \texttt{IVF}**
(note that they have the same cluster structure), our implementation of {\CHENGB K-means} clustering is based on the Faiss library~\cite{johnson2019billion}.} 
Then during {\CHENG the} \underline{query phase}, all {\CHENG algorithms} are implemented in C++.
For a new query, we first transform the query vector with the Eigen library~\cite{eigenweb} for fast matrix multiplication when running \texttt{AKNN+} and \texttt{AKNN}++ algorithms 
{\JIANYANGREVISION (For \texttt{AKNN}, \texttt{AKNN}* and \texttt{AKNN}**, they involve no transformation). Then we feed the vector into 
the \texttt{AKNN}, \texttt{AKNN+},  \texttt{AKNN++}, \texttt{AKNN}* and \texttt{AKNN}**
algorithms.}
Following~\cite{graphbenchmark, li2019approximate}, we disable all hardware-specific optimizations including SIMD, memory prefetching and multi-threading ({\JIANYANG including those in the Eigen library}) {\CHENG so as to focus on the comparison among algorithms themselves}. 


\smallskip
\noindent
\textbf{Parameter Setting.} 
{\JIANYANG For \texttt{HNSW}, two parameters are preset to control {\CHENG the construction of the graph}, namely $M$ to control the number of connected neighbors and $efConstruction$ to control the quality of approximate nearest neighbors. We follow the parameter settings of its original work~\cite{malkov2018efficient} where the parameters are set as $M=16$ and $efConstruction=500$.} 
For \texttt{IVF}, as suggested in the Faiss library~\footnote{\url{https://github.com/facebookresearch/faiss/wiki/Guidelines-to-choose-an-index}}, the number of clusters should be around the square root of the cardinality of the database. Since we focus on million-scale datasets, it's set to be {\CHENG 4,096}. 
For \texttt{ADSampling}, 
we vary $\epsilon_0$ {\CHENG with values of 
1.5, 1.8, 2.1, 2.4, 2.7 and 3.0
and study its effects in Section~\ref{subsub:parameter}. Based on the results, we adopt the setting of 
$\epsilon_0 = 2.1$
as the default one}.
{\CHENG 
{\JIANYANG Recall that in \texttt{ADSampling}, it 
{\JIANYANGREVISION incrementally samples} some dimensions of a data vector and performs a hypothesis testing in iterations. }
To avoid the overhead of 
frequent hypothesis testings, {\JIANYANG we sample $\Delta_d$ dimensions at each iteration. By default, we set $\Delta_d=32$.}
{\JIANYANG Its parameter study is also provided in Section~\ref{subsub:parameter}.}
}

All C++ source codes are complied by g++ 9.4.0 with 
\texttt{-O3} optimization under
Ubuntu 20.04LTS. The Python source codes (which are used during the index phase) are run on Python 3.8. All experiments are {\CHENG conducted} on a machine with AMD Threadripper PRO 3955WX 3.9GHz 16C/32T processor and 64GB RAM. The code and datasets are available at \url{https://github.com/gaoj0017/ADSampling}.


\begin{figure*}[ht]
  \centering 
    \includegraphics[width=17cm]{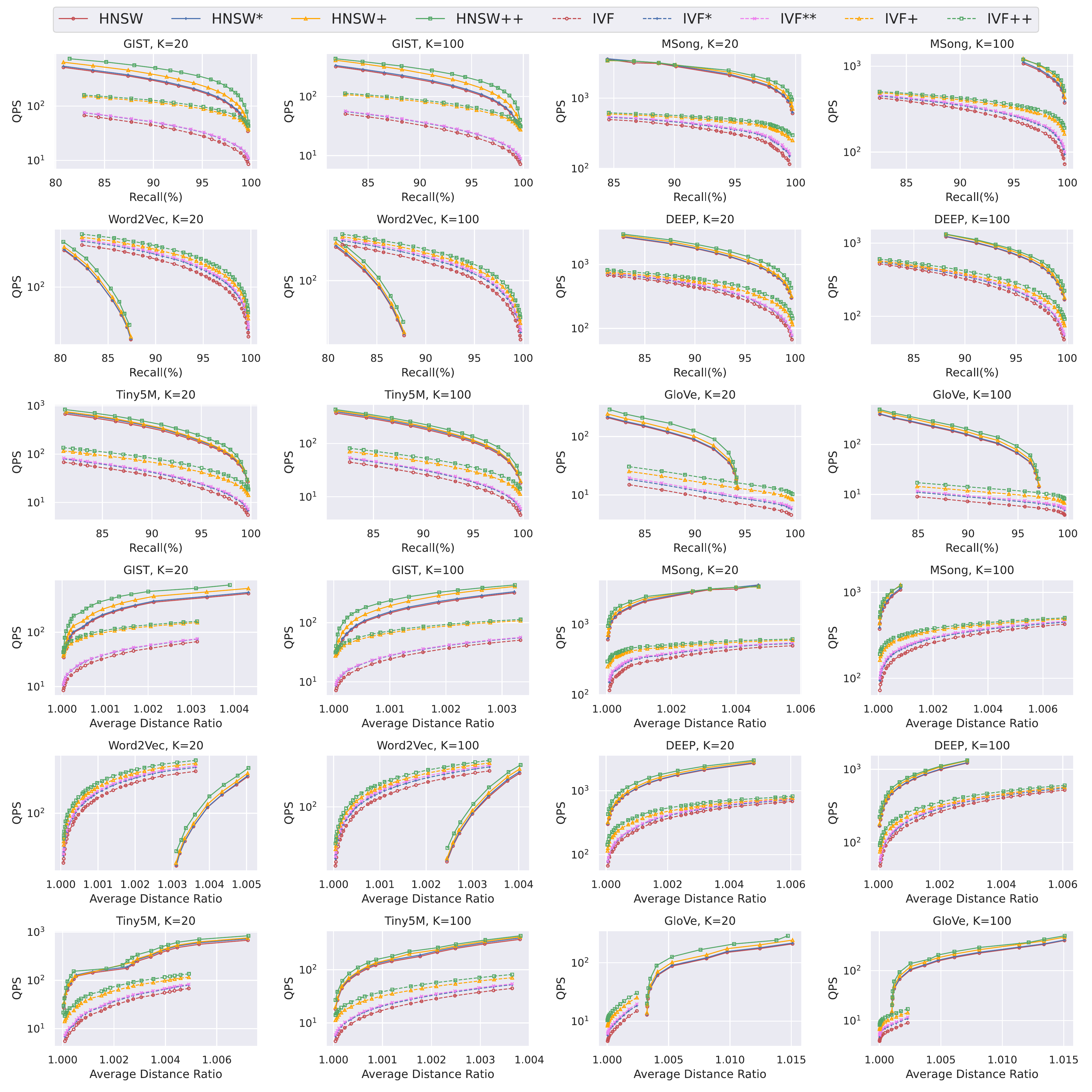}
  \vspace*{-4mm}
  \caption{{\JIANYANGREVISION Time-Accuracy Tradeoff (\texttt{HNSW} and \texttt{IVF}).}}
  \label{figure:time-accuracy}
\end{figure*}

\begin{figure*}[ht]

   \captionsetup[subfigure]{aboveskip=-5pt}
\begin{subfigure}[b]{0.48\linewidth}
    \includegraphics[width=\textwidth]{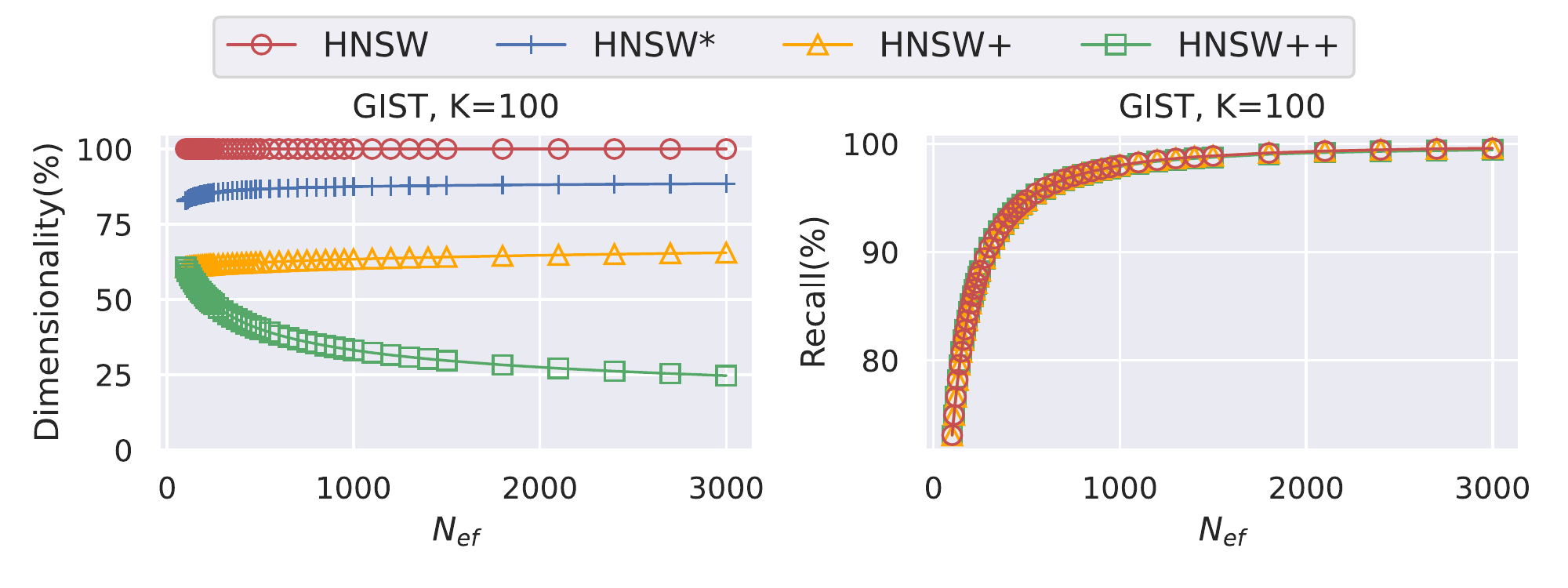}
    \caption{\texttt{HNSW}}
    \label{figure:evaluated_dimension HNSW}
\end{subfigure}  
\begin{subfigure}[b]{0.48\linewidth}
    \includegraphics[width=\textwidth]{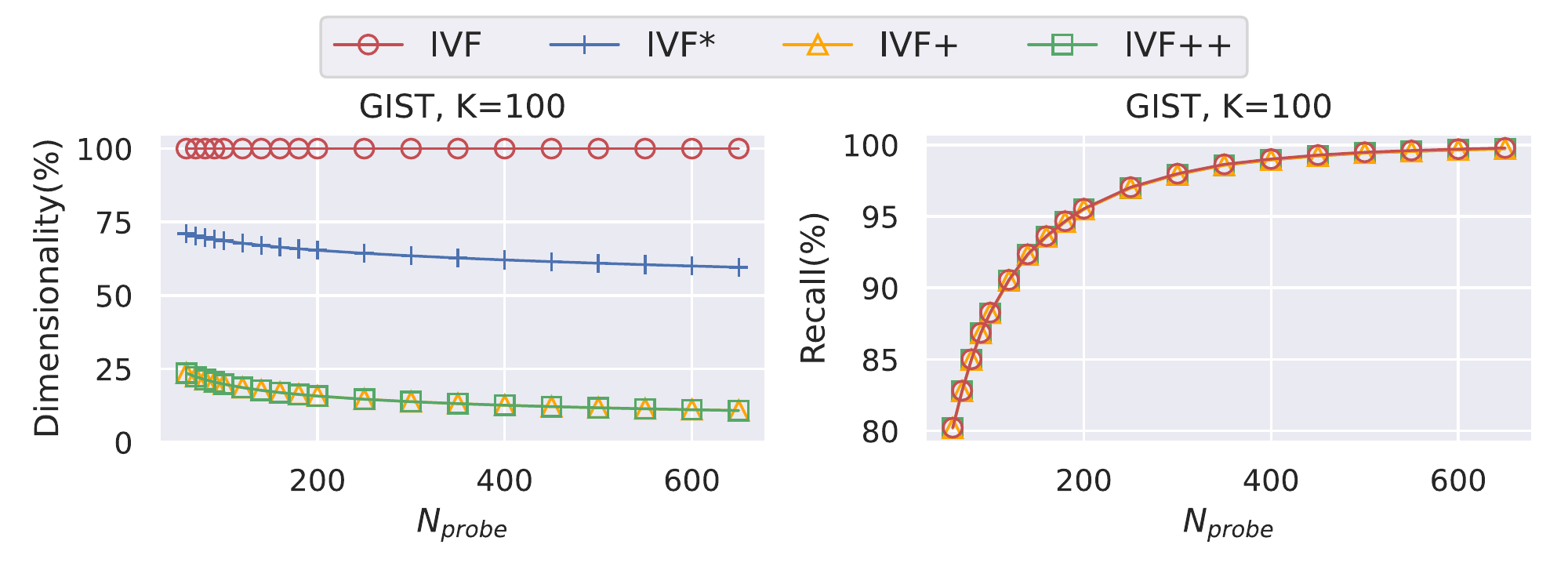}
    \caption{\texttt{IVF}}
    \label{figure:evaluated_dimension IVF}
\end{subfigure} 
\vspace*{-4mm}
\caption{\JIANYANGREVISION Evaluated Dimensionality and Accuracy.}
\vspace*{-4mm}
\label{figure:evaluated_dimension}
\end{figure*}

\subsection{{\CHENG Experimental Results}}
\label{section:time-accuracy}
\subsubsection{\textbf{{\CHENG Overall Results (Time-Accuracy Trade-Off)}}}
\label{subsec:main result}
{\JIANYANGREVISION We {\CHENG plot} the QPS-recall curve (upper panels, upper-right is better) and the QPS-ratio curve (lower panels, upper-left is better) by varying $N_{ef}$ for \texttt{HNSW}/\texttt{HNSW}*/\texttt{HNSW+}/\texttt{HNSW++} and $N_{probe}$ for \texttt{IVF}/\texttt{IVF}*/\texttt{IVF}**/\texttt{IVF+}/\texttt{IVF++} in Figure~\ref{figure:time-accuracy}.} 
We focus only on the region with the recall at least 80\% {\CHENG based on} practical needs.
Overall, with the results in Figure~\ref{figure:time-accuracy}, we can observe clearly 
%
{\JIANYANG that (1) the \texttt{AKNN+} algorithms (represented by the orange curves) outperform the plain \texttt{AKNN} algorithms (represented by the red curves), (2) the \texttt{AKNN++} algorithms (represented by the green curves) further outperform the \texttt{AKNN+} algorithms}, {\JIANYANGREVISION (3) the baseline method \texttt{HNSW}* (represented by the blue curves) brings very minor improvements on \texttt{HNSW} for all the tested datasets and (4) the baseline methods \texttt{IVF}* {\chengr (represented by the blue curves)} and \texttt{IVF}** (represented by the violet curves) are outperformed by \texttt{IVF+} consistently and significantly ({\chengr and by} \texttt{IVF++} {\chengr with an} even larger margin).}

{\JIANYANG
Besides, we have the following observations. (1) Our techniques bring more improvements on \texttt{IVF} than on \texttt{HNSW} (even when \texttt{IVF} performs better than \texttt{HNSW}, e.g., on Word2Vec). We ascribe it to the fact that other computations than DCOs of \texttt{HNSW} are heavier than {\CHENG those} of \texttt{IVF} (as shown in Figure~\ref{fig:cost statistics}). (2) Our techniques in general bring more improvements on high accuracy region than on low accuracy region (e.g., GIST 95\% v.s. 85\%). This is because {\CHENGB when} an AKNN algorithm targets higher accuracy, it unavoidably generates more low-quality candidates with larger distance gap $\alpha$, 
{\CHENGB for which it needs fewer dimensions for reliable DCOs.} {\JIANYANGREVISION (3) The data layout optimization brings more improvements on \texttt{IVF+} (i.e., \texttt{IVF++} v.s. \texttt{IVF+}) than on \texttt{IVF}* (i.e., \texttt{IVF}** v.s. \texttt{IVF}*). This is because \texttt{ADSampling} has the logarithmic complexity while the baseline \texttt{PDScanning} has the linear complexity. Specifically, the first $\Delta_d$ dimensions are sufficient for many DCOs when using \texttt{ADSampling} and thus, many accesses to the second array $A_2$ in Figure~\ref{fig:data layout ARScan} can be avoided. When using \texttt{PDScanning} for the DCOs, it will still access the second array frequently because it needs {\chengr more than $\Delta_d$} dimensions.}
}




\subsubsection{\textbf{{\CHENG Results of Evaluated Dimensions and Recall}}}
\label{subsub:dimensions and recall}

{\JIANYANG 
We then study the number of evaluated dimensions and the recall of 
{\JIANYANGREVISION \texttt{AKNN}/\texttt{AKNN+}/\texttt{AKNN++}/\texttt{AKNN}* (\texttt{AKNN}** has exactly the same curve as \texttt{AKNN}* and thus, is omitted)}
under the same search parameter setting ($N_{ef}$ for \texttt{HNSW} and $N_{probe}$ for \texttt{IVF}). {\CHENG For the number of evaluated dimensions, we measure its ratio over that of 
\texttt{AKNN} in percentage for the ease of comparison. The results are shown in Figure~\ref{figure:evaluated_dimension}.}
}

\smallskip
\noindent
\textbf{{\CHENG Overall Results.}} 
{\JIANYANG
In Figure~\ref{figure:evaluated_dimension}, we can observe clearly that \texttt{AKNN+} and \texttt{AKNN++} {\CHENG evaluate} much fewer dimensions than \texttt{AKNN} while {\CHENG reaching} nearly the same recall. Specifically, on GIST, for all tested values of $N_{ef}$, the accuracy loss of \texttt{HNSW+} and \texttt{HNSW++} (compared with \texttt{HNSW}) is no more than 0.14\% and that of \texttt{IVF+} and \texttt{IVF++} is no more than 0.1\%. At the same time,
\texttt{HNSW++} saves from 39.4\% to 75.3\% of the total dimensions, 
\texttt{HNSW+} saves from 34.5\% to 39.4\% and 
\texttt{IVF+}/\texttt{IVF++} save from 76.5\% to 89.2\%. {\JIANYANGREVISION The baseline method \texttt{HNSW}* saves from 10.9\% to 15.7\%, which explains its minor improvement on \texttt{HNSW}. \texttt{IVF}*/\texttt{IVF}** saves from 28.9\% to 40.4\%.} 
}

\begin{figure}[thb]
  \centering 
  \includegraphics[width=\linewidth]{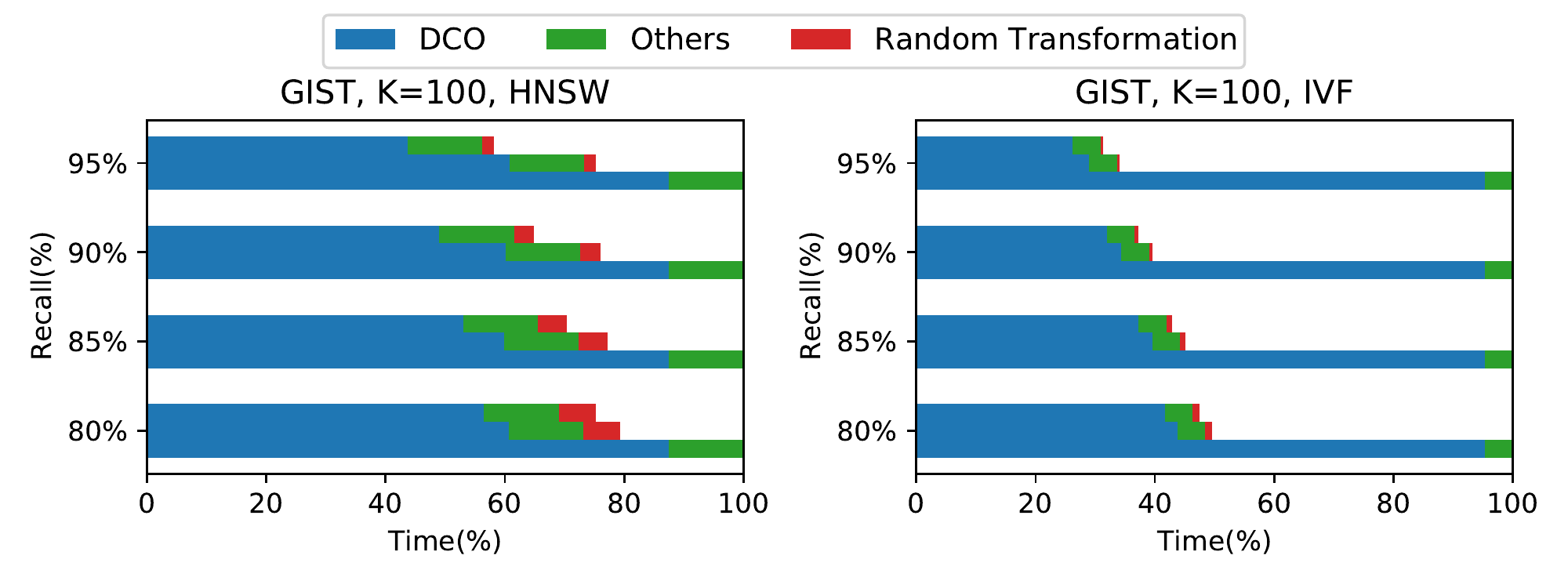}
  \vspace{-8mm}
  \caption{{\JIANYANGREVISION Decomposition of Time Cost (At a particular {\chengr recall}, the three horizontal bars from top to bottom, represent the \texttt{AKNN++}, \texttt{AKNN+} and \texttt{AKNN} algorithms, {\chengr respectively}. The time cost is normalized by the cost of the original \texttt{AKNN} algorithms).}}
  \vspace{-4mm}
  \label{fig:decomposition_simple}
\end{figure}

\smallskip
\noindent
\textbf{{\JIANYANG \texttt{HNSW+} v.s. \texttt{HNSW++}}.}
{\JIANYANG
We further compare \texttt{HNSW+} and \texttt{HNSW++}. According to Figure~\ref{figure:evaluated_dimension HNSW}, we have the following observations.
(1) \texttt{HNSW++} evaluates \emph{fewer dimensions} than \texttt{HNSW+}, 
{\CHENG which largely explains the result that \texttt{HNSW++} runs faster than \texttt{HNSW+}.}
(2) \texttt{HNSW++} reaches \emph{nearly the same recall} as \texttt{HNSW+}, which empirically shows that using approximate distances for graph routing has nearly the same effectiveness as using exact distances. 
(3) The evaluated dimensions (its ratio over those of \texttt{HNSW} in percentage) of \texttt{HNSW+} increases wrt $N_{ef}$ while those of \texttt{HNSW++} decreases wrt $N_{ef}$. This is because \texttt{HNSW+} conducts DCOs with the $N_{ef}^{th}$ NN's distance as the threshold, whose distance increases wrt $N_{ef}$, and thus a larger $N_{ef}$ leads to smaller distance gap $\alpha$, which entails more dimensions for a reliable DCO. 
For \texttt{HNSW++},
as mentioned in \ref{subsec:main result}, when targeting high recall, an AKNN algorithm inevitably generates many low-quality candidates with larger $\alpha$'s, and thus, it needs fewer dimensions for DCOs.
}

\smallskip
\noindent
\textbf{{\JIANYANG \texttt{IVF+} v.s. \texttt{IVF++}}.} \texttt{IVF+} and \texttt{IVF++} differ only in data layout, and thus they have exactly the same accuracy and evaluated dimensions. 



{\JIANYANGREVISION
\subsubsection{\textbf{Results of Time Cost Decomposition}}
\label{subsec: time decompostion}
We note that applying \texttt{ADSampling} entails the extra cost of randomly transforming the data and query vectors. In particular, the cost of transforming the data vectors lies in the index phase and can be amortized by all the subsequent queries on the same database. The transformation of the query vectors is conducted during the query phase when a query comes and its cost can be amortized by all the DCOs involved for answering the same query. 
{\chengr We implement this step {\JIANYANGREVISION (a.k.a, Johnson-Lindenstrauss Transformation~\cite{johnson1984extensions, JL_survey})} as a matrix multiplication operation for simplicity, which takes $O(D^2)$ time. We note that this step can be performed in less time with advanced algorithms~\cite{JL_survey}, e.g., it takes $O(D\log D)$ time with Kac's Walk~\cite{kac_walk}.}
We show the results of time cost decomposition on the dataset GIST. It has the highest dimensionality and correspondingly the largest overhead for {\chengr random transformation}. 
We decompose the time cost in 
Figure~\ref{fig:decomposition_simple}.
We note that the cost of {\chengr random transformation} for \texttt{HNSW+}/\texttt{HNSW++} takes at most 6.18\% of the total cost of the original \texttt{HNSW}. As the accuracy increases, the percentage decreases (e.g., for 95\% recall, it reduces to 2.02\%). For \texttt{IVF+}, the percentage is no greater than 1.11\%.

}


\subsubsection{\textbf{Results for {\CHENGB Evaluating} the Feasibility {\CHENGB of Distance Approximation Techniques} for Reliable DCOs}}
\label{subsubsec:reliable-dco}
\begin{figure}[thb]
  \centering 
   \includegraphics[width=\linewidth]{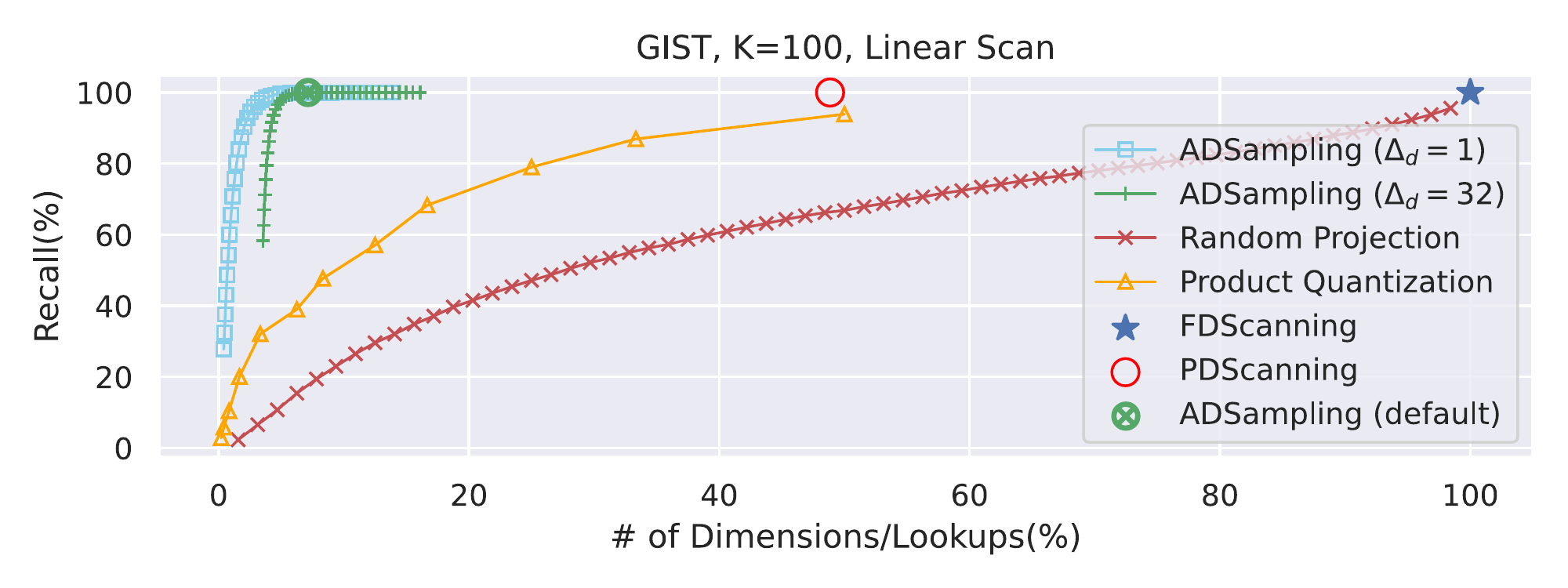}
  \vspace{-8mm}
  \caption{{\JIANYANGREVISION Feasibility for Reliable DCOs (Recall).}}
  \vspace{-4mm}
  \label{figure:feasibility}
\end{figure}

\begin{figure}[thb]
    \centering
    \includegraphics[width=\linewidth]{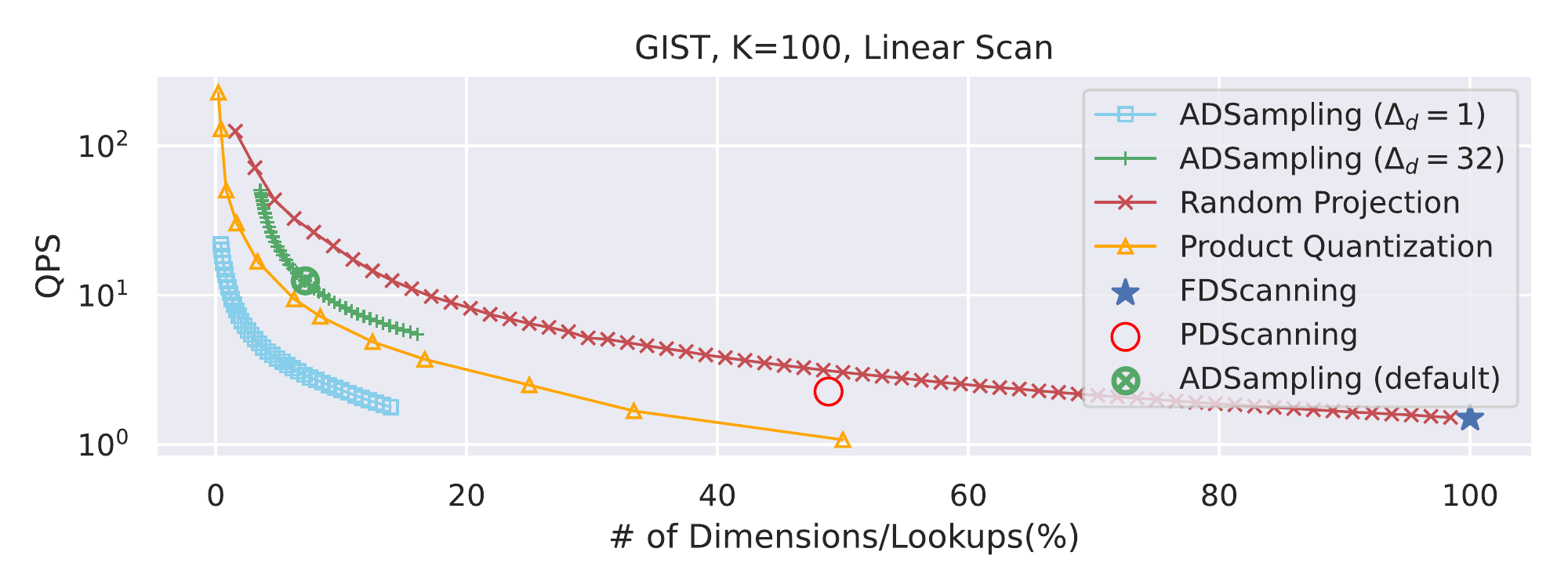}
    \vspace{-8mm}
    \caption{{\JIANYANGREVISION Feasibility for Reliable DCOs (QPS).}}
    \vspace{-4mm}
    \label{fig:qps linear}
\end{figure}

We next study the feasibility of 
{\CHENGB two distance approximation methods, including random projection and product quantization~\cite{jegou2010product} {\JIANYANGREVISION (with the typical setting of 256 centroids per partition~\cite{jegou2010product, learningtohashsurvey})},} for reliable DCOs. {\CHENGB We include the results of \texttt{ADSampling}, {\JIANYANGREVISION \texttt{PDScanning}} and \texttt{FDScanning} for comparison.} To test the best possible recall a method can reach, we conduct this experiment with an exact KNN algorithm, namely linear scan. Specifically, for random projection and product quantization, we scan all the data objects and return the K objects with the minimum approximate distances. For \texttt{ADSampling} and {\JIANYANGREVISION \texttt{PDScanning}}, like \texttt{IVF}, we maintain a KNN set and conduct DCOs for each object sequentially. We plot the recall-number of dimensions/lookups
~\footnote{
For product quantization, it refers to the quantization code size,
where evaluating each code would look up a table in memory (i.e., access memory randomly). For other methods, it refers to the number of dimensions, where evaluating each dimension applies some arithmetic computations. Note that they are not directly comparable because in modern CPUs, the former is much slower than the latter. 

} 
curves in Figure~\ref{figure:feasibility}. For \emph{random projection}, we vary the dimensionality of the projected vectors and observe that (1) it introduces 3.71\% accuracy loss {\CHENGB while} reducing only 1.04\% dimensions (2) {\CHENGB when} reducing half of the dimensionality, its recall is no more than 70\%. For \emph{product quantization}, we vary its quantization code size {(i.e., the number of partitions)} and observe that 
in the best possible case
{\JIANYANGREVISION (i.e., the case of encoding every two dimensions with one code)},
it still introduces 6.2\% accuracy loss. 
{\CHENGC Therefore, neither product quantization nor random projection can achieve reliable DCOs with remarkably better efficiency than \texttt{FDScanning}.}

For \texttt{ADSampling}, we test two settings $\Delta_d=1$ and $\Delta_d=32$. The former represents the best possible recall-dimension tradeoff of our method and the latter represents a practical setting with less frequent hypothesis testing (i.e., our default setting). We plot their curves by varying $\epsilon_0$ from 0.0 to 4.0. We observe that for $\Delta_d=1$, it samples 6.61\% of the total dimensions {\CHENGB while reaching} >99.9\% recall and for $\Delta_d=32$, it samples 7.11\% of the total dimensions {\CHENGB while reaching} > 99.9\% recall. Thus, \texttt{ADSampling} achieves much better recall-dimension tradeoff than \texttt{FDScanning}. 

\begin{figure}[thb]
  \centering 
  \includegraphics[width=\linewidth]{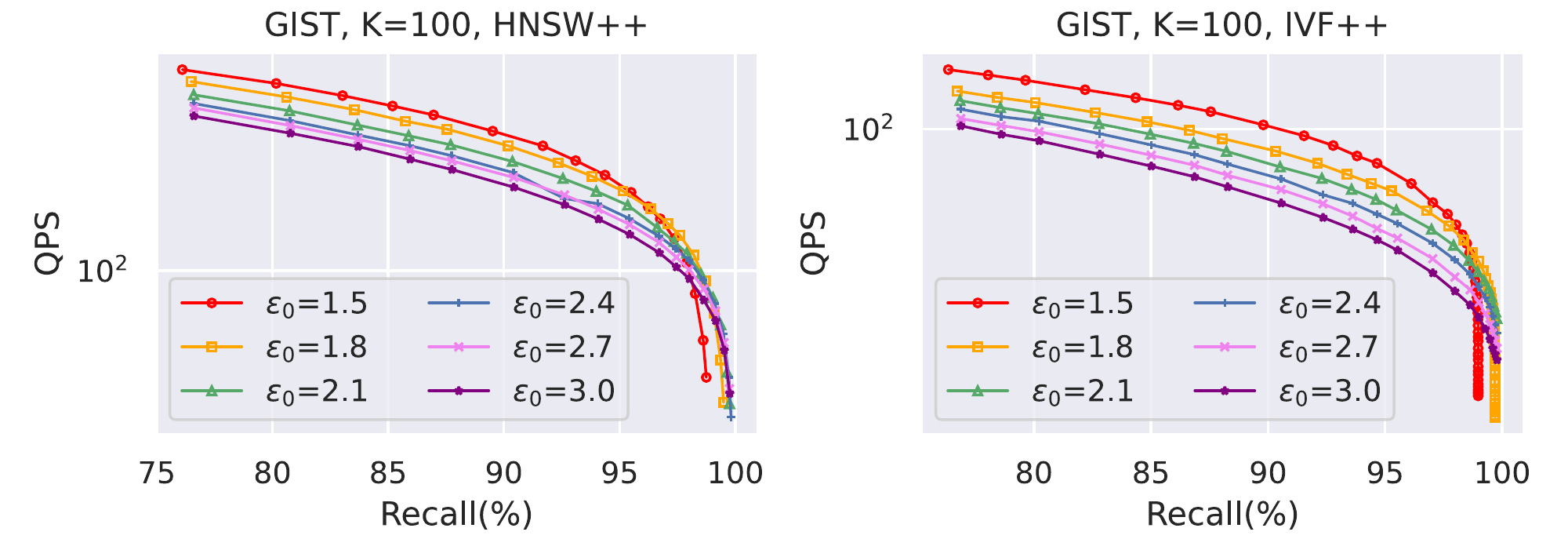}
  \vspace{-8mm}
  \caption{Parameter Study on $\epsilon_0$ of \texttt{AKNN++} Algorithms.}
  \vspace{-4mm}
  \label{figure:parameter}
\end{figure}

\begin{figure}[thb]
  \centering 
  \includegraphics[width=\linewidth]{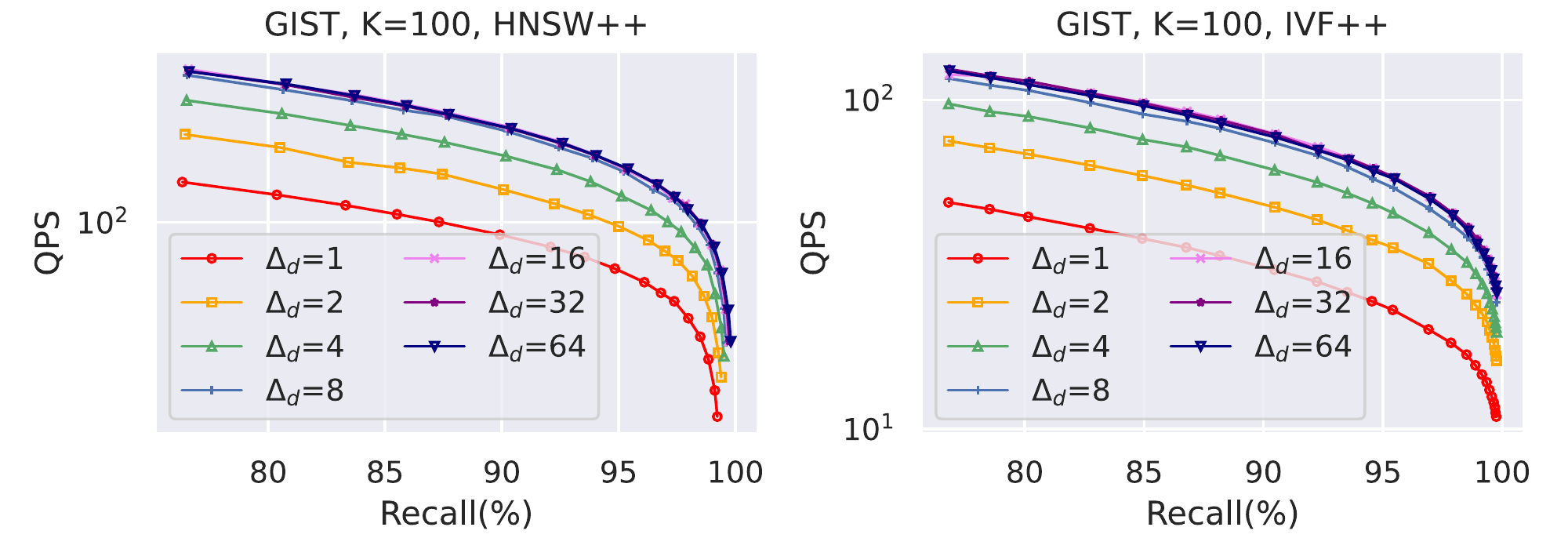}
  \vspace{-8mm}
  \caption{Parameter Study on the $\Delta_d$ of \texttt{AKNN++} Algorithms.}
  \vspace{-4mm}
  \label{figure:parameter gap}
\end{figure}

\begin{figure}[thb]
  \centering 
  \includegraphics[width=\linewidth]{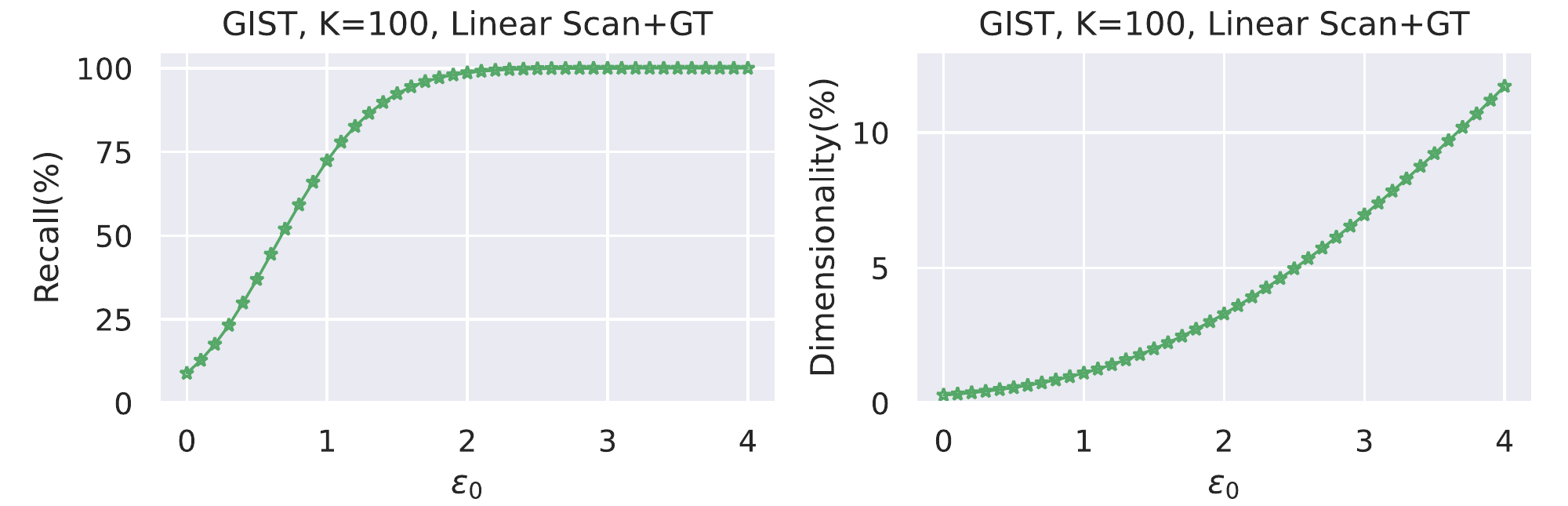}
  \vspace{-8mm}
  \caption{Verification for Theoretical Analysis.}
  \vspace{-4mm}
  \label{figure:verification}
\end{figure}

{\JIANYANGREVISION 
In Figure~\ref{fig:qps linear}, we plot the QPS-dimensions/lookups curves. We have the following observations. 
(1) {\chengr \texttt{ADSampling} (with default setting), which is marked with a green cross within a green circle, has the QPS significantly higher than \texttt{FDScanning} and \texttt{PDScanning}. This is because it exploits only 7.11\% of the total dimensions while achieving a recall over $99.9\%$.} 
(2) At the same dimensionality, random projection has its efficiency better than \texttt{ADSampling}. This is because random projection has fixed dimensionality and can organize the projected vectors sequentially in an array to achieve better cache-friendliness. However, we note that when random projection has the same QPS as \texttt{ADSampling} (default), its recall does not exceed 40\%. 

}

\subsubsection{\textbf{{\CHENG Results of} Parameter Study}}
\label{subsub:parameter}



{\JIANYANG
%

{\CHENG Parameter $\epsilon_0$ is a critical parameter for the \texttt{ADSampling} algorithm since it directly controls the trade-off between the accuracy and the efficiency (recall that a larger $\epsilon_0$ means a smaller significance value for the hypothesis testing, which further implies a more accurate result of the hypothesis testing).}
Figure~\ref{figure:parameter} plots the QPS-recall curves of \texttt{HNSW++} (left panel) and \texttt{IVF++} (right panel) with different $\epsilon_0$. 
In general, we {\CHENGC observe} from the figures that with {\CHENG a larger} $\epsilon_0$, the QPS-recall curves moves lower right. This is because {\CHENG a larger} $\epsilon_0$ leads to better accuracy at the cost of efficiency.
{\JIANYANGB We observe that when $\epsilon_0=2.1$, it introduces little accuracy loss while further increasing $\epsilon_0$ would decrease the efficiency. Thus, in order to improve the efficiency without losing much accuracy, we suggest to set $\epsilon_0$ around $2.1$.}
%
The results for \texttt{HNSW+} and \texttt{IVF+} are similar and omitted due to the page limit.

Figure~\ref{figure:parameter gap} plots the QPS-recall curves of \texttt{HNSW++} and \texttt{IVF++} with different $\Delta_d$. 
We observe that too frequent hypothesis testing (e.g., when $\Delta_d=1$) would do harm to the performance. 
It's worth noting that a small $\Delta_d$ implies that it can terminate sampling immediately when enough information is collected, {\CHENGB but it would} require more arithmetic operations for hypothesis testing. Our empirical study shows that when $\Delta_d=16, 32, 64$, it achieves the best {\CHENGC trade-off.}
}

{\JIANYANG
\smallskip
\noindent
\subsubsection{\textbf{{\CHENG Results for Verifying Theoretical Results}}}
\label{subsubsec:theoretical results}

We further empirically verify Lemma~\ref{theorem:ADSampling accuracy} and ~\ref{theorem:ADSampling efficiency}. 
As a verification study, the experimental setting is different. 
To eliminate the accuracy loss caused by AKNN algorithms, we conduct the verification study based on linear scan, which itself is an exact KNN algorithm.
{\JIANYANGB  Note that in KNN query processing, the result of the former DCOs can affect the distance thresholds of the latter DCOs, which introduces some bias into a verification study. To eliminate it, we provide a fixed distance threshold (the exact distance of the ground truth Kth NN) to make the DCOs independent with each other.}
To test the actual needed dimensionality, we set $\Delta_d=1$.
{\JIANYANGB With this setting, the recall represents the proportion of successful DCOs for positive objects (recall that those for negative objects will never fail), and thus, it empirically reflects the success probability of a single {\CHENGC DCO with} \texttt{ADSampling}.}
The left panel of Figure~\ref{figure:verification} shows that the failure probability indeed decays following a quadratic-exponential trend and reaches near 100\% accuracy around $\epsilon_0=2$. The right panel shows that {\CHENG the number of evaluated dimensions} increases following a quadratic trend, which is slow when $\epsilon_0$ is small (when $\epsilon_0=2$, the total {\CHENG number of} evaluated dimensions {\CHENG is} around $3\%$ of {\CHENG that of} the plain \texttt{FDScanning}). 

}

\section{Related Work}
\label{sec:related work}

\noindent
\textbf{Approximate K Nearest Neighbor Search.} 
Existing AKNN algorithms can be categorized into four types: 
{\JIANYANGREVISION (1) graph-based~\cite{malkov2018efficient, NSW, li2019approximate, fu2019fast, fu2021high, SISAP_graph}}, (2) quantization-based~\cite{jegou2010product, ge2013optimized, guo2020accelerating, ITQ, additivePQ, imi}, (3) {\JIANYANGREVISION tree-based ~\cite{muja2014scalable, dasgupta2008random, ram2019revisiting, beygelzimer2006cover, reviewer_M_tree}} and (4) hashing-based~\cite{indyk1998approximate, datar2004locality, c2lsh, tao2010efficient, huang2015query, sun2014srs, lu2020vhp, zheng2020pm, james_cheng}.
In particular, graph-based methods show superior performance for in-memory AKNN query. 
Quantization-based methods are powerful when memory is limited. Hashing-based methods provide rigorous theoretical guarantee. 
We refer readers to recent tutorials~\cite{tutorialThemis, tutorialXiao}, {\JIANYANGREVISION reviews and benchmarks~\cite{li2019approximate, annbenchmark, SISAP_benchmark, reviewer_paper, graphbenchmark}} for details. 
{\JIANYANG
There are also plentiful {\CHENG studies, which apply}
machine learning (ML) to accelerate AKNN~\cite{learning2route2019ml, reinforcement2route, adaptive2020ml, Dong2020Learning}. \cite{learning2route2019ml, reinforcement2route} apply reinforcement learning in graph routing which substitutes greedy beam search. \cite{adaptive2020ml} learns to early terminate searching. \cite{Dong2020Learning} uses ML to construct an index structure. Note that all above methods apply ML for candidate generation. 
{\CHENG These ML-based methods are orthogonal to our techniques} and our techniques can {\JIANYANGB help them with} 
finding KNNs among the generated candidates.
}

\smallskip
\noindent
\textbf{Random Projection {\CHENGB for AKNN}.} 
{\CHENGB While random projection can hardly be used for reliable DCOs during the phase of re-ranking candidates of KNNs as explained and verified earlier, it has been} widely applied in LSH~\cite{indyk1998approximate, datar2004locality, c2lsh, tao2010efficient, sun2014srs, huang2015query, lu2020vhp} and random partition/projection tree~\cite{ram2019revisiting, dasgupta2008random} {\CHENGB during  the phase of generating candidates of} KNNs. 
{\CHENGB Our study differs from these studies in (1) we project different objects to vectors with different dimensions flexibly while these studies project all objects to vectors with equal dimensions; (2) we set the number of dimensions of a projected vector for an object automatically based on its DCO via hypothesis testing while these studies need to set the number with manual efforts; and (3) we use the projected vectors (in DCOs) during the phase of 
{\JIANYANGB finding out KNNs from the generated}
candidates while these studies use the projected vectors during the phase of generating candidates. Therefore, these studies are orthogonal to our study.}

{\JIANYANGREVISION 
\smallskip
\noindent
\textbf{Dimension Sampling {\CHENGB for AKNN}.} 
We notice that a MAB (multi-armed bandit)-oriented approach~\cite{dimension_sampling_aistat} also applies dimension sampling and claims logarithmic complexity. Our study is different from \cite{dimension_sampling_aistat} in the following aspects. Problem-wise, \cite{dimension_sampling_aistat} targets the AKNN problem itself and aims to find a superset of the set containing the KNNs. It is non-trivial to adapt \cite{dimension_sampling_aistat} to DCOs (the focus of our paper). Theory-wise, \cite{dimension_sampling_aistat}'s logarithmic complexity relies on some strong assumptions on the data (which may not hold in practice) while ours relies on no assumptions. Technique-wise, (1) \cite{dimension_sampling_aistat} samples the original vectors directly while ours first randomly transforms the vectors and then samples the transformed vectors. Our way has the advantage that the error bound of an approximate distance is based on the concentration inequality of random projection and does not rely on any assumptions as \cite{dimension_sampling_aistat} does; and (2) \cite{dimension_sampling_aistat} uses some lower/upper bounds to determine the number of sampled dimensions while ours uses sequential hypothesis testing. Our way has no false positives while \cite{dimension_sampling_aistat} has both false positives and false negatives. In summary, \cite{dimension_sampling_aistat} and our work only share a high-level idea of dimension sampling and differ in many aspects including problem, theory and technique.
}
\section{Conclusion and Discussion}
\label{sec:conclusion}


We identify the distance comparison operation which dominates the time cost of {\CHENGB nearly all} AKNN algorithms and demonstrate opportunities to improve its efficiency. 
We propose a new randomized algorithm for the DCO which runs in logarithmic time wrt $D$ {\CHENGB in most cases} and succeeds with high probability. Based on it, we further develop one {\chengf generic} and two algorithm-specific techniques for AKNN algorithms. 
Our experiments show that the enhanced AKNN algorithms outperform the original ones consistently. 
We also provide rigorous theoretical analysis for all our techniques. 

We would like to highlight the following extensions and applications of our techniques.
(1) Our techniques can be trivially extended to two other {\chengf widely}-adopted similarity metrics, {\chengf namely} cosine similarity and inner product, via simple transformation. 
Specifically, the cosine-based similarity search on some given data and query vectors is equivalent to the Euclidean nearest neighbor search on their normalized data and query vectors where \texttt{ADSampling} is applicable.  
The inner product comparison of whether $\left< \mathbf{o}, \mathbf{q}  \right> \ge r$ can be reduced to the DCO of whether $\left\| \mathbf{o} / \|\mathbf{o} \|- \mathbf{q} / \|\mathbf{q} \| \right\|  \le \sqrt {2 - 2r / (\| \mathbf{o}\| \cdot \|\mathbf{q}\| )} $, where the distance threshold equals to $\sqrt {2 - 2r / (\| \mathbf{o}\| \cdot \|\mathbf{q}\| )} $~\footnote{It can be verified as follows: $\left< \mathbf{o}, \mathbf{q}  \right> \ge r \iff  \left< \mathbf{o} / \|\mathbf{o} \|, \mathbf{q} / \|\mathbf{q} \|\right> \ge r / (\| \mathbf{o}\| \cdot \|\mathbf{q}\| ) $
$\iff \| \mathbf{o} / \|\mathbf{o} \| \|^2 - 2\left< \mathbf{o} / \|\mathbf{o} \|, \mathbf{q} / \|\mathbf{q} \|\right> + \| \mathbf{q} / \|\mathbf{q} \|\|^2 \le 2 - 2r / (\| \mathbf{o}\| \cdot \|\mathbf{q}\| ) \\ \iff \left\| \mathbf{o} / \|\mathbf{o} \|- \mathbf{q} / \|\mathbf{q} \| \right\| ^2 \le 2 - 2r / (\| \mathbf{o}\| \cdot \|\mathbf{q}\| )$.}.
(2) DCOs are also ubiquitous in many other {\chengf tasks} of high-dimensional data management and analysis such as clustering~\cite{kmeans} and outlier detection~\cite{outlier_detection}. Our techniques have the potential to accelerate existing methods for those {\chengf tasks} by reducing the cost of DCOs while keeping the accuracy. 



\section{Acknowledgements}
This research is supported by the Ministry of Education, Singapore, under its Academic Research Fund (Tier 2 Awards MOE-T2EP20220-0011 and MOE-T2EP20221-0013). Any opinions, findings and conclusions or recommendations expressed in this material are those of the author(s) and do not reflect the views of the Ministry of Education, Singapore. We would like to thank Yi Li (SPMS, NTU) for answering many questions about high-dimensional probability and the anonymous reviewers for providing constructive feedback and valuable suggestions.

\bibliographystyle{ACM-Reference-Format}
\bibliography{sample-base}

\appendix
\section*{appendix}

\section{Theoretical Analysis}
\label{section:theory}

In this section, we prove Lemma~\ref{theorem:ADSampling efficiency} in detail. 
We assume that {\CHENGB we sample one additional dimension of $\mathbf{y}$ each time.}
Let $\gamma(d) = (1 + \epsilon_0 / \sqrt {d} )$. 
We have
\begingroup
\allowdisplaybreaks
\begin{align}
    \mathbb{E} \left[ \hat D \right]  
    =& \sum_{d=1}^{D} d\cdot \mathbb{P} \left\{ \hat D = d \right\}
    = \sum_{d=1}^{D} \mathbb{P} \left\{ \hat D \ge d \right\}
    \\=& \sum_{d=1}^{D} \mathbb{P} \left\{  \forall p<d,  \sqrt {\frac{D}{p}}  \| \mathbf{y}|_{1,2,...,p}\| \le \gamma(p) \cdot r  \right\}      \label{reduction:interpret}
    \\=& \sum_{d=1}^{D} \mathbb{P} \left\{  \forall p<d, \sqrt {\frac{D}{p}} \| \mathbf{y}|_{1,2,...,p}\| \le \gamma(p) \cdot \frac{\| \mathbf{y}\| }{1+\alpha} \right\} 
    \\\le& 1 + \sum_{d=1}^{D-1} \mathbb{P} \left\{ \sqrt {\frac{D}{d}} \| \mathbf{y}|_{1,2,...,d}\| \le \gamma(d) \cdot \frac{\| \mathbf{y}\| }{1+\alpha}  \right\}    \label{reduction: relax hypothesis testing}  
\end{align}
\endgroup
where (\ref{reduction:interpret}) is because $\hat D \ge d$ {\CHENGB represents the event that} all the previous hypothesis testings cannot reject the hypothesis and (\ref{reduction: relax hypothesis testing}) relaxes the event corresponding to all the testings (i.e., $\forall p < d$) to that corresponding to the last testing (i.e., $p=d-1$). 

We denote $\tilde d:= \epsilon _{0}^{2} / \alpha^2, d_0 := \mathrm{ceil} (\tilde d) $ and relax the probability for $d \le d_0$ to $1$:
\begin{align}
    \mathbb{E} \left[ \hat D \right]  \le  1 + d_0   + \sum_{d=d_0 +1 }^{D} \mathbb{P} \left\{ \sqrt {\frac{D}{d}} \| \mathbf{y}|_{1,2,...,d}\| \le \frac{\gamma (d)}{1+\alpha } \| \mathbf{y} \| \right\}  \label{reduction:separate analyze}
\end{align}
Let's focus on the last term of (\ref{reduction:separate analyze}) and deduce from it as follows,
\begingroup
\allowdisplaybreaks
\begin{align}
    &\sum_{d=d_0 +1 }^{D} \mathbb{P} \left\{ \sqrt {\frac{D}{d}} \| \mathbf{y}|_{1,2,...,d}\| \le \frac{\gamma (d)}{1+\alpha } \| \mathbf{y} \| \right\}   \\
    =&\sum_{d=d_0 +1}^{D} \mathbb{P} \left\{ \sqrt {\frac{D}{d}} \| \mathbf{y}|_{1,2,...,d}\| \le \left[ 1- \left( 1-\frac{\gamma(d)}{1+\alpha}  \right)  \right]   \| \mathbf{y} \| \right\}    \label{reduction:rewrite}
    \\\le& \sum_{d=d_0+1}^{D} \exp \left[ -c_0 d \left( 1 - \frac{\gamma(d)}{1+\alpha}  \right)^2 \right]  
    \label{reduction:lemma3.1}
    \\=&\sum_{d=d_0+1}^{D} \exp \left[ -\frac{c_0 \alpha^2}{(1+\alpha)^2}\left( \sqrt {d} - \sqrt{\tilde{d}}  \right) ^2  \right]  \label{reduction:expand and sort out}
    \\\le& \int_{d_0}^{D} \exp \left[ -\frac{c_0 \alpha^2}{(1+\alpha)^2}\left( \sqrt {x} - \sqrt {\tilde{d}}  \right) ^2  \right] \mathrm{d} x  \label{reduction:integration}
    \\=& \int_{d_0}^{D} \exp \left[ -\frac{c_0 \epsilon_0^2}{(1+\alpha)^2}\left( \sqrt {\frac{x}{\tilde{d}} } - 1 \right) ^2  \right] \mathrm{d} x  
    \\=&\tilde{d} \int_{d_0 /\tilde d}^{D/\tilde{d}} \exp \left[ - \frac{c_0 \epsilon_0^2}{(1+\alpha)^2} \left( \sqrt {u} -1 \right)^2  \right]  \mathrm{d} u  
    \label{reduction:substitute u}
    \\\le& \tilde{d} \int_{1}^{+\infty} \exp \left[ - \frac{c_0 \epsilon_0^2}{(1+\alpha)^2} \left( \sqrt {u} -1 \right)^2  \right]  \mathrm{d} u \label{reduction:relaxing to inf}
\end{align}
\endgroup
{\CHENGC where
(\ref{reduction:rewrite}) rewrites it to fit the format of Lemma~\ref{eq:concen}, 
(\ref{reduction:lemma3.1}) applies Lemma~\ref{lemma:concen}, 
(\ref{reduction:expand and sort out}) plugs in $\gamma(d)$, 
(\ref{reduction:integration}) relaxes (\ref{reduction:expand and sort out}) to an integration,
(\ref{reduction:substitute u}) substitutes $u = x / \tilde d$, and
(\ref{reduction:relaxing to inf}) relaxes the integration to $[1,+\infty)$.}
{\JIANYANG 
Next we first analyze the case of $\alpha \le \epsilon_0$ as follows.
}
\begin{align}
    \text{(\ref{reduction:relaxing to inf})}\le& \tilde{d} \int_1^{+\infty} \exp \left[ - \frac{c_0 \epsilon_0^2}{(1+\epsilon_0)^2} \left( \sqrt {u} -1 \right)^2  \right]  \mathrm{d} u  \label{reduction:alpha <= eps}
    \\\le& \tilde{d} \int_1^{+\infty} \exp \left[ - \frac{c_0}{4} \left( \sqrt {u} -1 \right)^2  \right]  \mathrm{d} u \label{reduction:eps >= 1}
\end{align}
{\CHENGC where (\ref{reduction:alpha <= eps}) is because $\alpha \le \epsilon_0$ and (\ref{reduction:eps >= 1}) is yielded when setting $\epsilon_0 \ge 1$ for reasonable accuracy.}
Note that the integration is convergent so as to be bounded by a constant. For the case of $\alpha \le \epsilon_0$, we have 
\begin{align}
    \mathbb{E} \left[ \hat D  \right]  = 1 + d_0 + O(\tilde{d}) = O(\tilde{d}) = O \left( \alpha^{-2} \cdot \epsilon_0^2 \right)  \label{eq:quadratic}
\end{align}
{\JIANYANG
{\CHENGB For the case of} $\alpha > \epsilon_0$, its expected dimensionality must be no greater than the case of $\alpha =\epsilon_0$ because its distance gap $\alpha$ is larger. Thus, its expected dimensionality is upper bounded by $O(1)$.
}
{\JIANYANGREVISION
\section{Results of Tree-based and Hashing-based Methods}
\label{appendix:section tree and hashing}
\begin{figure*}[thb]
  \centering 
    \includegraphics[width=17cm]{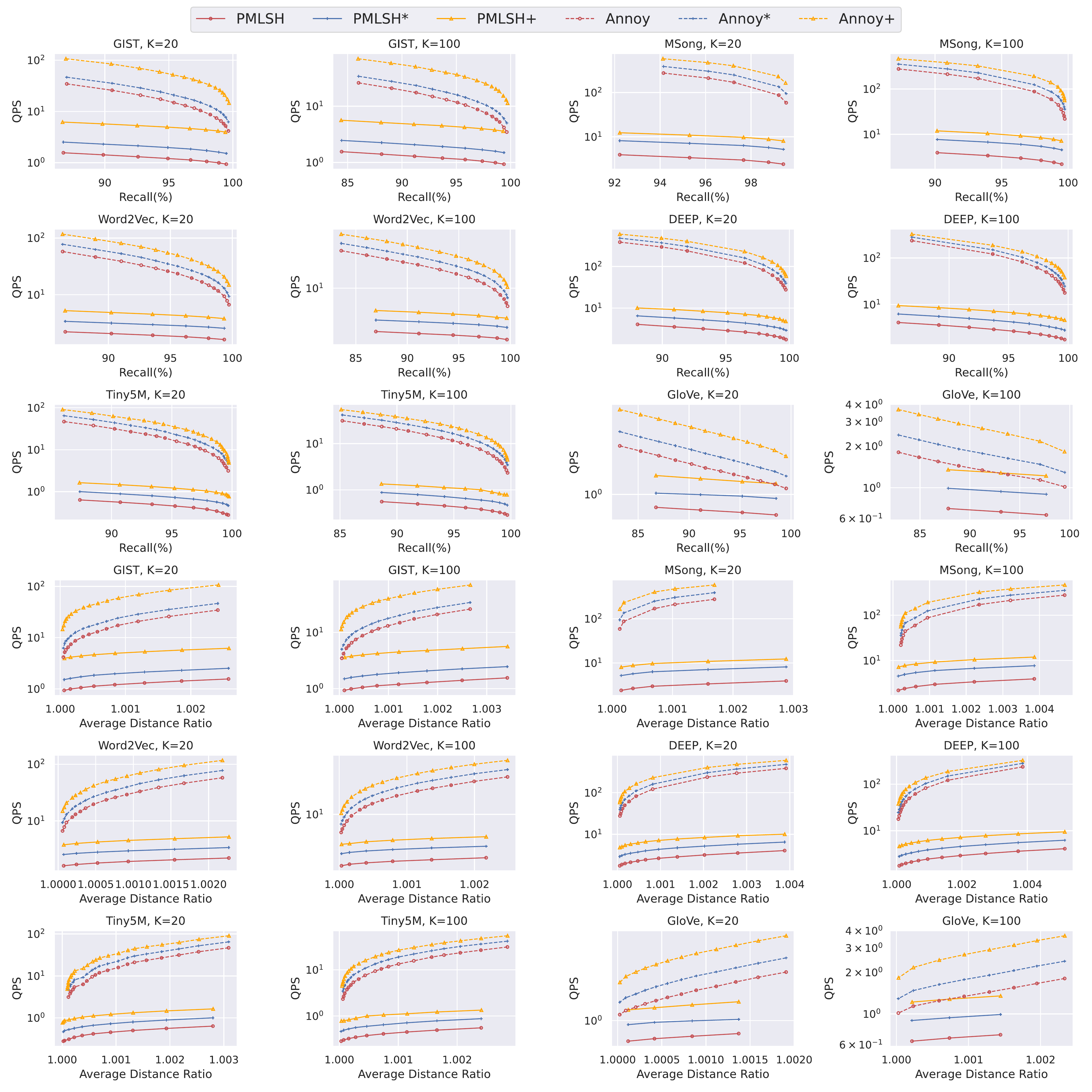}
  \vspace*{-4mm}
  \caption{{\JIANYANGREVISION Time-Accuracy Tradeoff (\texttt{PMLSH} and \texttt{Annoy}).}}
  \vspace*{-4mm}
  \label{figure:time-accuracy-tree-and-hashing}
\end{figure*}
For \texttt{Annoy}, following \cite{li2019approximate}, we set the number of trees $N_{tree} = 50$. 
{\JIANYANGREVISION During the \underline{index phase}, we feed the raw data vectors into the indexing algorithm of \texttt{Annoy} (note that \texttt{Annoy}, \texttt{Annoy+} and \texttt{Annoy}* have the same index structure). Then during the \underline{query phase}, for \texttt{Annoy}/\texttt{Annoy}*, we load the index and the raw data vectors into main memory, generate candidates by feeding the raw query vector into the the query algorithm of \texttt{Annoy} and re-rank the candidates with \texttt{FDScanning}/\texttt{PDScanning}. For \texttt{Annoy+}, we load the index and the transformed data vectors into main memory, generate candidates by feeding the raw query vector into the the query algorithm of \texttt{Annoy} and re-rank the candidates with \texttt{ADSampling}. }
For \texttt{PMLSH}, following \cite{zheng2020pm}, we set the dimensionality of random projection as 15, the size of internal and leaf nodes of the PM-Tree as 16. Similar to \texttt{Annoy}, during the \underline{index phase}, we build the indexes based on the raw data vectors. During the \underline{query phase}, we generate candidates by feeding the raw query vectors to the search algorithm of \texttt{PMLSH} and re-rank them with \texttt{FDScanning}, \texttt{PDScanning} and \texttt{ADSampling} based on raw vectors, raw vectors and transformed vectors respectively. 
For both methods, we vary the number of accessed candidates to control the time-accuracy tradeoff.  
{\chengr We exclude the optimization of data layout for this experiment since it is not applicable for index ensembles (e.g., tree ensembles of \texttt{Annoy}).}
We plot the QPS-recall and QPS-average distance ratio curves {\chengr of the compared algorithms} in Figure~\ref{figure:time-accuracy-tree-and-hashing}. It shows that \texttt{AKNN+} outperforms the \texttt{AKNN}* and \texttt{AKNN} algorithms consistently and significantly.

}

\end{document}